\documentclass[12pt]{article}
\usepackage{amssymb,amsmath,latexsym}
\usepackage{graphicx}
\usepackage{xcolor}
\usepackage{graphicx}
\usepackage{calrsfs}
\usepackage[utf8]{inputenc}
\usepackage[english]{babel}
\usepackage{amsmath}
\usepackage{bbold}
\usepackage{amsthm}
\usepackage{graphicx}
\usepackage{subcaption}
\usepackage[cal=cm]{mathalfa}
\usepackage{booktabs} 
\usepackage{hyperref}
\usepackage[ruled,vlined]{algorithm2e}
\usepackage{hyperref}
\newcommand{\footremember}[2]{%
    \footnote{#2}
    \newcounter{#1}
    \setcounter{#1}{\value{footnote}}%
}
\newcommand{\footrecall}[1]{%
    \footnotemark[\value{#1}]%
}

\newtheorem{lemma}{Lemma}

\newtheorem{proposition}{Proposition}

\newtheorem{remark}{Remark}
\title{Games of Social Distancing during an Epidemic: Local vs Statistical Information}
\author{A.-R. Lagos\footremember{trailer}{ National Technical University of Athens, School of Electrical and Computer Engineering, 9 Iroon Polytechniou str.,Athens, Postal Code 157 80, Greece.   \newline E-mails: A.-R. Lagos~ lagosth993@gmail.com,  I.~ Kordonis jkordonis1920@yahoo.com,   G.P. ~Papavassilopoulos yorgos@netmode.ntua.gr}
\and I. Kordonis\footrecall{trailer}
\and G. P. Papavassilopoulos\footrecall{trailer} }

\begin{document}
\maketitle

\begin{abstract}
The spontaneous behavioral changes of the agents during an epidemic can have significant effects on the delay and the prevalence of its spread. In this work, we study a social distancing game among the agents of a population, who determine their social interactions during the spread of an epidemic. The interconnections between the agents are modeled by a network and local interactions are considered. The payoffs of the agents depend on their benefits from their social interactions, as well as on the costs to their health due to their possible contamination. The information available to the agents during the decision making plays a crucial role in our model. We examine two extreme cases. In the first case, the agents know exactly the health states of their neighbors and in the second they have statistical information for the global prevalence of the epidemic. The Nash equilibria of the games are studied and, interestingly, in the second case the equilibrium strategies for an agent are either full isolation or no social distancing at all. Experimental studies are presented through simulations, where we observe that in the first case of perfect local information the agents can affect significantly the prevalence of the epidemic with low cost for their sociability, while in the second case they have to pay the burden of not being well informed. Moreover, the effects of the information quality (fake news), the health care system capacity and the network structure are discussed and relevant simulations are provided, which indicate that these parameters affect the size, the peak and the start of the outbreak, as well as the possibility of a second outbreak.
\end{abstract}

\section{Introduction}
The emergence of the Covid-19 pandemic is one of the most significant events of this era. It affects many sectors of human daily life, it indicates the inefficiency of many health care systems and it leads to state interventions in the functioning of the society through urgent measures, to economic depression and to human behavioral changes. Different states followed significantly different strategies to contain the pandemic and achieved respectively different levels of success. However, as the pandemic progresses, the interest about its nature, its dynamics and the need to control it made epidemiology a scientific field known to almost everyone and its terminology used daily by the media and included in many conversations. Humans spontaneously react to the emergence of Covid-19, following or disrespecting the state directions and legislation, adaptively adjusting their behavior based on their perceived risk.  Thus, a question naturally arises: How this spontaneous behavioral change of humans affects the prevalence of the disease and under what assumptions would it be effective in reducing the spread of the outbreak? \par
Humankind has always been haunted by epidemics, some of which have been recorded from historians, such as the plague in ancient Athens (430BC) and the Black Death in medieval Europe. So, epidemiology has concerned a lot of scientists during the ages and mathematical models for this field were first developed in the $18^{th}$ century \cite{Bernoulli}. Nowadays, the most prevalent approach in epidemic modeling is the compartmental models, introduced a century ago \cite{Kermack},\cite{Ross}. These models assume that there exist several compartments where an agent can belong (e.g. Susceptible-Infected-Recovered) and derive ordinary differential equations for the description of the dynamics of the population in each compartment. A main assumption for that analysis to hold is the well mixing of the population. However, there is enough evidence from social and other kinds of human networks that this assumption does not hold in many cases. \par
Due to that fact, novel approaches in epidemic modeling take into consideration the heterogeneous networked structure of human interconnections \cite{Pastor-Satorras}. A branch of these approaches uses results from the percolation theory to estimate the spread of the epidemic \cite{Newman1,Newman2,Meyers,Garnett,Sander}. Another branch, that is gaining a lot of attention \cite{Epstein1}, is the agent-based models \cite{Epstein2,Cliff-Australia}, which consider several parameters of each agent profile (e.g. residence, age, mobility pattern) and run computer simulations for large populations of such agents to estimate the spread of the disease. There exist also several recent works \cite{Zhang1,Chang,Bagnoli} which take into consideration the networked structure of human interconnections and they derive the -compartmental- models they use, through a mean-field approach. \par
Regardless of its derivation and its mathematical formulation, the usefulness of epidemic modeling is to guide states and/or individuals in taking the right protective measures to contain the epidemics. These measures, besides the efforts to develop appropriate meditation, can be roughly organized into two categories: vaccination \cite{Zhang1,Chang}, \cite{Bauch1,Bauch2,Reluga1,Reluga2}, \cite{Zhang2,Fine-Clarkson} and behavioral changes
\cite{Kremer,Vardavas,Del_Valle,Chen2,Funk-Review}, \cite{Reluga3,Poletti1,Poletti2,Poletti3}, \cite{Funk1,Chen1,d'Onofrio}. In the second case, the actions taken by the agents may vary from usage of face masks and practice of better hygiene to voluntary quarantine, avoidance of congregated places, application of preventive medicine and other safe social interactions. \par
In both cases, a very important fact that determines the effectiveness of the protective measures is that the agents make rational choices with regard to the self-protective activities they adopt by comparing the costs and benefits of these actions. Even in the case that a central authority imposes a policy, it is often up to agents to fully comply with this or not, even if they will have to pay a high cost if they get caught. From these considerations game theory arises as a natural tool to model and analyze the agents behavior with respect to the adoption of protective measures. Many recent studies on this field incorporate a game theoretic analysis \cite{Bauch1,Bauch2,Reluga1,Reluga2,Reluga3,Poletti1,Poletti2,Poletti3,Fu,Zhang2,van_Boven,Fine-Clarkson,Philipson1,Philipson2,Philipson3}, some of which are summarized in \cite{Game_survey}. It should be pointed out here, that the assumption of rational agents does not always hold true, since in many cases the agents decisions are not based on the maximization of their personal utility. Moreover, in the cases it holds it is a ``double-edged sword" \cite{Zhang2}, because self-interest leads the agents to adopt strategies different than the ones which maximize group interest \cite{van_Boven,Fine-Clarkson,Philipson1,Philipson2,Philipson3}. Another main characteristic of game-theoretic approaches is the crucial role of information available to the agents for their decision. The remarkable impact of information on the epidemic outbreaks has been pointed out in \cite{Funk1,Chen1,d'Onofrio}, where the authors consider an extra dynamic modeling the spread of information, coupled with the contagion dynamics. The informed agents are supposed to alter their behavior and affect this way the disease prevalence. \par
Following the research directions presented in the previous paragraphs, and specifically the game-theoretic approaches for the modeling of behavioral changes \cite{Reluga3,Poletti1,Poletti2,Poletti3},\cite{Chen1}, we propose and analyze a game-theoretic model for social distancing in the presence of an epidemic. Our model differs from \cite{Reluga3,Poletti1,Poletti2,Poletti3},\cite{Chen1} since it takes under consideration the networked structure of human interconnections and the locality of interactions, without attempting a mean-field approach. Each agent is considered to have her own state variables and information and choosing her action based on these - so it could be characterized an agent-based approach. Moreover, the actions of the agents affect the intensity of their relations with their neighbors and use or do not use the available connections. Changes in the topology of the network have been considered as a phenomenon in \cite{Gross,Shaw,Zanette}, but not from a game-theoretic perspective where the agents can choose rationally which connections to use and induce this way an ``active" topology. Furthermore, there exist several works on game-theoretic models which consider the networked structure of human interconnections, such as \cite{Zhang2,d'Onofrio,Fu,Somarakis}, where the strategy adoption is based on imitation of ones neighbors. Contrary to that, in our model the agents do not imitate the most effective strategy of their neighborhood, but design their best response based on the available information. We consider two different information patterns: perfect local information for the states of ones neighbors and statistical information for the global prevalence of the epidemic and investigate the different effects of these patterns.\par
Through the analysis of the proposed model we get several results. At first, we observe that in the case of perfect local information the agents can affect significantly the prevalence of the epidemic with low cost for their sociability, while in the case of statistical information they have to pay the burden of not being well informed. Secondly, in the case the agents have only statistical information, each agent's action is either full isolation or no social distancing at all. Lastly, we investigate, through experimental studies, the effects of the information quality (fake or biased news), the health care system capacity and the network structure and we conclude that these parameters affect the size, the peak and the start of the outbreak, as well as the possibility of a second outbreak.\par
The rest of the paper is organized as follows. In section \ref{s.model} the model for the epidemic outbreak and for the social distancing game between the agents is introduced. In section \ref{s.full_info} we analyze the game for the case that the agents have perfect local information for the states of their neighbors. In section \ref{s.distr_info} we analyze the game for the case that the agents have statistical information for the global prevalence of the epidemic. In section \ref{simulation_parameters} we present simulations for the games with the two different information patterns and compare the results. A discussion follows in section \ref{discussion}, where several variations of the problem are considered, such as experimentation on various network types \cite{random}-\cite{Scale_free}, the impact of fake information and of the finite capacity of a health care system and related simulations are presented and annotated.

\section{The model}\label{s.model}

We denote by $G=(V,E)$ an undirected graph, where $V=\{1,...,n\}$ is the set of its nodes representing the agents and $E\subset V\times V$ is the set of its edges indicating the social relations between the agents. $A=\{a_{ij}\}$ is the adjacency matrix of the graph i.e., $a_{ij}=1$ if $(i,j)\in E$, otherwise $a_{ij}=0$. $N_i=\{j: (i,j)\in E\}$ is the neighborhood of agent $i$, and  $\bar{N_i}=N_i \cup \{i\}$. $d_i=\sum_{j\in N_i}a_{ij}$ is the degree of node $i$, that is the number of her neighbors. We consider also a matrix $S=\{s_{ij}\}$, with the same sparsity pattern with the adjacency matrix $A$, which indicates the desire of each agent to meet with each one of her neighbors.\par
Social distancing is one of the most effective behavioral changes that people can adopt during an epidemic outspread. However, as mentioned in the introduction, the choice to adopt this altered behavior is, in many cases, up to the agents. So, we consider a social distancing game, which is repeated at each day during the outspread of the epidemic. The actions of the agents model the intensity of the relations with each one of their neighbors  they choose to have at each day. So, denoting by $k$ the current day, the action of  agent $i$ is a vector of length equal to the number of her neighbors given by:
\begin{equation}\label{strategies}
  u^i(k)=[u^i_{j_1}(k)...u^i_{j_{d_i}}(k)] \in [0,1]^{d_i},
\end{equation}
where:
\begin{equation*}
  N_i=\{j_1,...,j_{d_i}\}.
\end{equation*}

According to the strategies chosen by the agents we have an induced weighted adjacency matrix $W(k)=[w_{ij}(k)]$ for the network, which indicates the meeting probabilities between two neighbors at day $k$, where $w_{ij}(k)$ have the following form:

\begin{equation}\label{rep_weights}
w_{ij}(k)= \left\{
\begin{array}{ll}
      0 &,\textrm{if} \quad a_{ij}=0 \\
      u^i_j(k)u^j_i(k) &,\textrm{if} \quad a_{ij}=1 \\
\end{array}
\right.
\end{equation}

We consider that each agent has a health state consisted of two variables $x_i(k)$, which indicates if the agent has been infected before day $k$ and $r_i(k)$, which indicates the duration of her infection and consequently if she has recovered. Here we assume that all the infected agents recover after $R$ days.
\vskip 0.3 cm
The vector $x^0=[x_i^0]$ indicates the initial conditions for the $x_i$ state of the agents. The probability $p_x^0$ indicates the distribution of the initial conditions, which are i.i.d. random variables:
\begin{equation}\label{state_distr}
x_i^0= \left\{
\begin{array}{ll}
      0 &, \textrm{w.p.} \quad 1-p_x^0\\
      1 &, \textrm{w.p.} \quad p_x^0 \\
\end{array}
\right.
\end{equation}

The vector $r^0=\mathbb{0}_n$ indicates the initial conditions for the $r_i$ state of the agents.
\vskip 0.3 cm
These states evolve as follows:

\begin{equation}\label{state_rep}
x_i(k+1)= \left\{
\begin{array}{ll}
      x_i(k) &, \textrm{w.p.} \prod_{j \in N_i}(1-w_{ij}(k)p^cx_j(k)\mathcal{X}_{\{r_j(k)<R\}})\\
      1 &, otherwise \\
\end{array}
\right.
\end{equation}

\begin{equation}\label{costate_rep}
r_i(k+1)= \left\{
\begin{array}{ll}
      r_i(k)+x_i(k) &, \textrm{if} \quad  r_i(k)<R\\
      R &,\textrm{if} \quad r_i(k)=R  \\
\end{array}
\right.
\end{equation}
where $R$ is the duration of the recovery period.
\vskip 0.3cm
The probabilities $w_{ij}(k)p^cx_j(k)\mathcal{X}_{\{r_j(k)<R\}}$ indicate the possibility to have a meeting at day  $k$ and get infected by another agent. That agent can transmit the disease if she has been infected $(x_j(k)=1)$ and has not recovered yet $(r_j(k)<R))$, which is shown with the use of the characteristic function:

\begin{equation*}
\mathcal{X}_{\{r_j(k)<R\}}= \left\{
\begin{array}{ll}
      1 &, \textrm{if} \quad r_j(k)<R\\
      0 &, \textrm{if} \quad r_j(k)=R  \\
\end{array}
\right.
\end{equation*}

\begin{remark}\label{rem1}
  In this simple model, which is a discrete analogue of the  SIR model  on graphs, we assume that every infected agent recovers. That is to avoid changes in the graph topology, which would make the analysis of the game much more difficult. We expect this to cause minor differences in the case of an epidemic with low mortality.
\end{remark}

In order to model the probable contamination of an agent $j$ by her neighbor agent $i$, we make a similar assumption with the mean field approach \cite{Bagnoli},  where the authors assume that the graph topology has no loops and there is no correlation between the states of the agents. Thus, the contamination probability can be expressed as a function of the well known basic reproduction number $R_0$:
\begin{equation}\label{contamination_probability}
  p^c(R_0)=1-(1-\frac{R_0}{\bar{d}})^{\frac{1}{R}}.
\end{equation}
Similar derivations for the probabilities that govern the transmission of the disease over networks of interconnected agents are existing in the relevant bibliography, such as \cite{Newman1}.\par
We assume that the agents choose rationally their actions, based on the available information, by maximizing their payoffs. These payoffs are considered to  depend solely on the benefits from the social interactions between the agents and on the costs to their health due to possible contamination. In reality, the decision of a behavioral change depends also on socioeconomic and ethical considerations, which are omitted in this first approach, for the sake of simplicity. So, in our case the instantaneous payoffs depend on two terms. The first one indicates the satisfaction that each agent derives by the interaction with her neighbors, these benefits differ between her neighbors. The second term shows the costs an agent suffers if she has been infected. Since the agent does not know her health state the next day, she tries to estimate it based on the available information. The parameters $G_i$ indicate the importance of the infection for each agent. We divide the agents into two groups: the vulnerable (large $G_i$) and the ones who are non-vulnerable (small $G_i$). The game is in fact dynamic since the payoffs depend on the evolving health states of the agents. However, the agents are considered to be myopic and able to predict just a day ahead  so the payoffs have the following form at each day:

\begin{equation}\label{payoffs}
  J_i(k)=\sum_{j\in N_i}s_{ij}u^i_j(k)u^j_i(k)-G_iE\{x_i(k+1)|I_i(k)\}\mathcal{X}_{\{r_i(k+1)<R\}}
\end{equation}
where $I_i(k)$ is the information available to agent $i$ at day $k$. The agents decide what action to take based on this information. So, the actions are, in fact, strategies of the available information:

\begin{equation}\label{strategy_form}
  u^i(k)=\gamma_i(I_i(k))
\end{equation}

\section{Perfect local state feedback information}\label{s.full_info}
We first study  the case where the agents have perfect local state feedback information. That is, agents know exactly their current  health state and the current health states of their neighbors before taking the decision to meet them or not:

\begin{equation}\label{perfect_local_state_feedback_rep}
  I_i^F(k)=\{ x_j(k),r_j(k):  j \in \bar{N}_i\}.
\end{equation}

In order to analyze the social distancing game under the perfect local state feedback information \eqref{perfect_local_state_feedback_rep}, we follow a step-wise analysis, considering a static, one-step game. All the time indexes, indicating the days, will be omitted during this analysis. Based on the information \ref{perfect_local_state_feedback_rep} we can explicitly calculate the conditional expectation of each agent's next state $E\{x_i^+| I_i\}$:

\begin{equation*}\label{state_expectation}
  E\{x_i^+|I_i\}=x_i \prod_{j \in N_i}(1-w_{ij}p^cx_j\mathcal{X}_{\{r_j<R\}})+(1- \prod_{j \in N_i}(1-w_{ij}p^cx_j\mathcal{X}_{\{r_j<R\}}))
\end{equation*}
since from \eqref{strategy_form} the strategies are measurable on the sigma fields defined by $x$, so $E\{u^i_j|x\}=u^i_j$. Thus, the payoffs have the following form:

\begin{equation}\label{J_step_1}
  J_i=\sum_{j\in N_i}s_{ij}u^i_ju^j_i-\left[G_i(x_i-1) \prod_{j \in N_i}(1-w_{ij}p^cx_j\mathcal{X}_{\{r_j<R\}})+G_i\right]\mathcal{X}_{\{r_i^+<R\}}.
\end{equation}

\begin{proposition}
  The strategy profile $u=\mathbb{0}_{\sum d_i}$ is a Nash equilibrium for the game with perfect local state feedback, since it results to indifference for all the agents.
\end{proposition}
However, we can observe the existence of other Nash equilibria.

\begin{proposition}
  The best response of each agent always contains a point in $\{0,1\}^{d_i}$, i.e. the vertices of the action space. Therefore, there is no strict Nash equilibrium in $[0,1]^{\sum d_i}\setminus \{0,1\}^{\sum d_i}$
\end{proposition}

\begin{proof}
We calculate the first and second partial derivatives of $J_i$:
\begin{equation*}
  \frac{\partial J_i}{\partial u^i_j}=u^j_i\left[s_{ij}+G_i(x_i-1)\mathcal{X}_{\{r_i^+<R\}}p^cx_j\mathcal{X}_{\{r_j<R\}}\prod_{k \in N_i\setminus\{j\}}(1-u^i_ku^k_ip^cx_k\mathcal{X}_{\{r_k<R\}})\right]
\end{equation*}

\begin{equation*}
    \frac{\partial^2 J_i}{(\partial u^i_j)^2}=0
\end{equation*}
for all $j \in N_i$, so:
\begin{equation}
  \nabla^2J_i=0
\end{equation}
and thus $J_i$ is a harmonic function. So, form the maximum principle for harmonic functions on compact sets (\cite{renardy} chapter 4) we conclude that the local maxima of $J_i$ with respect to $u_i$ are on the boundary of $[0,1]^{d_i}$. Applying successively the maximum principle for the faces and the edges of the hypercube $[0,1]^{d_i}$, observing that $J_i$ is still harmonic on each face of the hypercube with respect to the free variables on that face (the $u^i_j$ that are not fixed to $0$ or $1$), we conclude that the best response of each agent always contains a point in $\{0,1\}^{d_i}$.
\end{proof}

\begin{remark}
  If agent $i$ is infected, $x_i=1$ and $r_i<R$, then $J_i=\sum_{j\in N_i}s_{ij}u^i_ju^j_i-G_i$ and if she has been recovered, $r_i=R$, it is assumed that she cannot get infected again. So, in these cases, an optimal strategy for her is $u^i_j=1$, $\forall j \in N_i$, since if $u^j_i=1 \implies u^i_j=1$ and if  $u^j_i=0$ she is indifferent so she can also choose  $u^i_j=1$.
\end{remark}

\begin{remark}
  If agent $i$ and agent $j$ are neighbors and agent $i$ is not infected ($x_i=0$) and agent $j$ is not infected ($x_j=0$) or recovered ($r_j=R$) the optimal strategies for their interaction are $u^i_j=1$ and $u^j_i=1$, since if $u^j_i=1$:  $J_i(u^i_j=1)-J_i(u^i_j=0)=s_{ij}>0$ and if $u^i_j=1$: $J_j(u^j_i=1)-J_j(u^j_i=0)=s_{ji}>0$ .
\end{remark}
So defining the following sets:

\begin{equation}\label{infected}
  \textrm{Infected}_i=\{j\in N_i: x_j=1, r_j<R\}
\end{equation}
and $|\textrm{Infected}_i|$ is the number of elements of $\textrm{Infected}_i$, we conclude that:

\begin{equation*}
  J_i=J_i(u^i_j: j\in  \textrm{Infected}_i),
\end{equation*}
since the rest strategies are fixed. In this case, the computation of the equilibrium strategies is a single objective, multi-variable, integer optimization problem for each agent, which can be solved easily using the following algorithm for each agent in $O(|\textrm{Infected}_i|(\log(|\textrm{Infected}_i|)+1))$ iterations:

\begin{algorithm}[H]\label{algorithm1}
\SetAlgoLined
\KwResult{The optimal strategies $(u^i_j)^*$ for $j\in  \textrm{Infected}_i$}
 Sort the parameters $s_{ij}$, $j\in  \textrm{Infected}_i$ in decreasing order\;
 Define the sequence of indices $j_1...j_{| \textrm{Infected}_i|}$ to be the j-indices of the previous ordering\;
 Define the strategies $\bar{u^i_0}=\mathbb{0}_{ \textrm{Infected}_i}$, $\bar{u^i}_k=\{u^i_{j_1}=1...u^i_{j_k}=1, u^i_{j_{k+1}}=0... u^i_{j_{| \textrm{Infected}_i|}=0}\}$, $k=1...| \textrm{Infected}_i|$\;
 $k=0$\;
 $\Delta J_i =1$\;

 \While{$\Delta J_i>0$ and $k \leq | \textrm{Infected}_i|$}{
  $\Delta J_i= s_{ij_k}-G_ip^c(1-p^c)^{k}$\;
  $k=k+1$\;
 }

 $(u^i_j)^*=\bar{u^i}_{k-1}(j_k=j)$\;
 \caption{Solution of the optimization problem for each agent}
\end{algorithm}

\begin{remark}
  The strategy profile $u^i_j=\max\{x_i, 1-x_j\}$ is a Nash equilibrium for the game with perfect local state feedback \eqref{perfect_local_state_feedback_rep}, if $\forall i \not \in \cup_i  \textrm{Infected}_i : \max\{s_{ij}: j\in  \textrm{Infected}_i\}<G_ip^c$
\end{remark}

This equilibrium shows the phenomenon that in the case the agents are highly vulnerable to the disease and they know the state of their neighbors, they communicate with all the healthy ones in order to maximize their payoffs and the infected try to communicate also with their neighbors for the same reason but they are banned by them. So, this equilibrium results to higher payoffs for the non infected agents:
\begin{equation}
J_i= \left\{
\begin{array}{ll}
      \sum_{j \in N_i}s_{ij}(1-x_j\mathcal{X}_{\{r_j<R\}}) &, x_i=0 \quad \textrm{or} \quad r_i=R  \\
      -G_i &, x_i=1 \quad \textrm{and} \quad r_i<R\\
\end{array}
\right.
\end{equation}

\section{Information for the distribution of the states}\label{s.distr_info}

The second case that we study is the case where the agents have statistical information for the distribution of the states, which in our case is a Bernoulli distribution, assuming that the agents ignore the correlations between their states. We assume also that all the agents know the same distribution with the same parameters and that they have no memory for the past values of these parameters:

\begin{equation}\label{info_distr_rep}
  I_i^D(k)=\{p_x(k), p_r(k)\},
\end{equation}
where

\begin{equation}\label{px_k}
  p_x(k)=\frac{|\{i:1\leq r_i(k)<R\}|}{N},
\end{equation}
is the percentage of ill agents at day $k$ and

\begin{equation}\label{px_k}
  p_r(k)=\frac{|\{i: r_i(k)=R\}|}{N},
\end{equation}
is the percentage of recovered agents at day $k$.\par
Furthermore, we assume that each agent chooses the same probability to meet each one of her neighbors and then makes $d_i$ random experiments to decide if she will meet each one of them.

\begin{equation}\label{random_str_rep}
  u^i_j(k)= \left\{
\begin{array}{ll}
      1 &, \textrm{w.p.}\quad p^i_u(k)\\
      0 &, \textrm{otherwise}
\end{array}
\right.
\end{equation}
this is rational only if the utility earned from each interaction is the same from all the neighbors of each agent: $s_{ij}=s_i$, $\forall j\in N_i$. We assume that this symmetry holds for this case. Consequently, the strategy space of each agent is:

\begin{equation}\label{1_dim_str_space_rep}
  p^i_u(k) \in [0,1].
\end{equation}
We then drop $k$ in order to proceed with the analysis of one step of the game. In order to study the equilibria of this game we have firstly to compute the expectation of the state of the agents based on the available information \eqref{info_distr_rep}. Thus, we compute at first the expectation of the next state of an agent given the current states:

\begin{equation*}
  E\{x_i^+|x, r\}=1-(1-x_i)\prod_{j\in N_i}(1-u^i_ju^j_ip^cx_j\mathcal{X}_{(r_j<R)}),
\end{equation*}
next we compute the expectation of the previous conditional expectation over all the states:

\begin{equation*}
 E_{x, r}\big\{ E\{x_i^+|x, r\}\big\}=1-(1-p_x)\prod_{j\in N_i}(1-u^i_ju^j_ip^cp_x(1-p_r)),
\end{equation*}
and thus the criteria have the following form:

\begin{equation*}
  J_i=s_i\sum_{j\in N_i}u^i_ju^j_i+\left[G_i(1-p_x)\prod_{j\in N_i}(1-u^i_ju^j_ip^cp_x(1-p_r))-G_i\right](1-p_r),
\end{equation*}
where the strategies are random and uniform for all the neighbors of an agent according to eq.\eqref{random_str_rep}, so we have to compute the expected criteria, given the probabilities of the uniform strategies:

\begin{equation}
  \hat{J}_i=E\{J_i|p^i_u,p^j_u,j\in N_i\}=s_ip^i_u\sum_{j\in N_i}p^j_u+\left[G_i(1-p_x^0)\prod_{j\in N_i}(1-p^i_up^j_up^cp_x^0)-G_i\right](1-p_r)
\end{equation}
Each agent wants to maximize $\hat{J}_i$ w.r.t. $p^i_u$. However, computing the first two derivatives of $\hat{J}_i$ w.r.t. $p^i_u$ we get:

\begin{equation*}
  \frac{\partial \hat{J}_i}{\partial p^i_u}=s_i\sum_{j\in N_i}p^j_u-G_i(1-p_x)(1-p_r)\sum_{j\in N_i}p^j_up^cp_x(1-p_r)\prod_{k\in N_i\setminus \{j\}}(1-p^i_up^k_up^cp_x(1-p_r)),
\end{equation*}
and

\begin{equation*}
  \frac{\partial^2 \hat{J}_i}{(\partial p^i_u)^2}=G_i(1-p_x)(1-p_r)\sum_{j\in N_i}p^j_up^cp_x(1-p_r)\sum_{k\in N_i\setminus \{j\}}p^k_up^cp_x(1-p_r)\prod_{l\in N_i\setminus \{j,k\}}(1-p^i_up^l_up^cp_x(1-p_r))\geq 0,
\end{equation*}
which indicate that each $\hat{J}_i$ is convex with respect to $p^i_u$, independently of the strategies $p^j_u$ of the other agents, so the possible equilibria are in $\{0,1\}^N$. In order to characterize the Nash equilibria of this game we observe that it is strategically equivalent to the following one:

\begin{equation}
  \tilde{J}_i(p_u^i, p_u^{-i})=a_ip_u^i\sum_{j \in N_i}p_u^j+\prod_{j \in N_i}(1-bp_u^ip_u^j),
\end{equation}
where:

\begin{align*}
  a_i = \frac{s_i}{G_i(1-p_x)(1-p_r)}, \quad b = p^cp_x(1-p_r),
\end{align*}
and

\begin{align*}
  p_u^i & \in \{0,1\}, \quad \forall i.
\end{align*}

We proceed with the calculation of the best response for each agent. If agent $i$ has $m_i$ of her neighbors playing $p_u^j=1$ her payoff is:
\begin{equation*}
  \tilde{J}_i(p_u^i,m_i)=a_im_ip_u^i+(1-bp_u^i)^{m_i}.
\end{equation*}
Thus:
\begin{align*}
  \tilde{J}_i(0,m_i) & = 1,\\
  \tilde{J}_i(1,m_i) & = a_im_i+(1-b)^{m_i}.
\end{align*}
We define the following functions:
\begin{align}\label{f1}
  f_i(m) = \tilde{J}_i(1,m) = a_im+(1-b)^m = a_im+e^{m \ln(1-b)}
\end{align}
The best response of each agent is:
\begin{equation}\label{BR_i}
  BR_i(m_i)= \left\{
  \begin{array}{ll}
      1 &, \textrm{if} \quad f_i(m_i)>1\\
      0 &, \textrm{otherwise} \\
  \end{array}
  \right.
\end{equation}
So, we have the following algorithm for the computation of the strategies corresponding to a Nash equilibrium:

\begin{algorithm}[H]\label{algorithm2}
\SetAlgoLined
\KwResult{The optimal strategies $p_u^{i*}$}
 Set $p_u^i=1$, $\forall i$ \\
 Compute $f_i(m_i)$, $\forall i$ ($m_i=d_i$)

\While{$\exists f_i(m_i)\leq1$}{
  \If{$f_i(m_i)\leq1$}{
    Set $p_u^i=0$
  }
  Compute new $m_i$, $\forall i$
  Compute new $f_i(m_i)$, $\forall i$
}

\caption{Computation of the NE strategies for the game with information for the distribution of the states }
\end{algorithm}

\begin{proposition}
    There exists a Nash equilibrium of the game with statistical information for the distribution of the states. Furthermore, Algorithm \ref{algorithm2} converges to the Nash equilibrium in $\mathcal{O}(N^2)$ steps.
\end{proposition}

\begin{proof}
To prove this proposition we firstly prove the following lemma:
\begin{lemma}\label{f_prop}
  For the functions $f_i(m)$, defined in \eqref{f1}, there exists a unique $m_0 \in \mathbb{R}_+$ such that $f(m_0)=1$ and for all $ m>m_0$, $m\in \mathbb{N}$ : $f(m)>1$
\end{lemma}

\begin{proof}
  It is easily observed that $f_i(m)$ is convex and $f_i(0)=1$ for each $i$.
  So, if $f'_i(0) \geq 0 \Rightarrow f_i(m)>1$, $\forall m$, in this case $m_0=0$.
  Else if $f'_i(0) \leq 0 \Rightarrow \exists ! m_0 \in \mathbb{R}_{+}^{*}$ : $f(m_0)=1$ and $\forall m>m_0$, $m\in \mathbb{N}$ : $f(m)>1$ due to the convexity of $f_i(m)$.
\end{proof}

Due to this lemma, beginning with the maximum feasible value for $m_i$ (which is $d_i$) the changes in the agents strategies from $1$ to $0$ can result only in the decrease of their neighbors $m_j$'s and thus it is possible to happen only one change of strategy ($1\rightarrow 0$) for each agent until $f_i(m_i)\geq 1$, $\forall i$, so the algorithm converges. Moreover, due to this observation, in the worst case the `while-loop' will run $N$ times and so the algorithm will converge in $\mathcal{O}(N^2)$ steps.
\newline
The point that the algorithm converges is a Nash equilibrium of the game, since the agents actions are their best responses to their active contacts numbers $m_i$'s and for this profile of $m_i$'s no agent will be benefited from a unilateral deviation from her action.\newline
Furthermore, we should point that, since the algorithm is in fact a descent on the possible $m_i$-profiles, i.e. it initializes with all the contacts being active ($m_i=d_i$, $\forall i$) and each $m_i$ decreases or stays the same, the Nash equilibrium that the algorithm converges is the one corresponding to the maximum possible sociability for the agents.
\end{proof}

\begin{remark}
  If for each agent $i$ it holds that $s_id_i+G_i(1-p_x)(1-p_r)[(1-p^cp_x(1-p_r))^{d_i}-1]>0$ then the strategy profile $p^i_u=1, \forall i$ is a N.E. of that game.
\end{remark}

\begin{proposition}
  The strategy profile $p^i_u=0, \forall i$ is again a N.E., since it results to indifference between the unilateral changes of each agent.
\end{proposition}

\section{Numerical studies}\label{simulation_parameters}

In this section we present several simulations for the social distancing games under the two different information structures in order to compare the disease prevalence and the agents payoffs in both cases, as well as the importance of some parameters of the model. For these simulations we consider a repeated version of this game. The payoffs of the agents in this case have the form \eqref{payoffs}, indicating the myopic behavior for the agents, who cannot predict the future consequences of their actions. The strategies considered in the following simulations are the Nash Equilibrium strategies of the static games of the previous sections repeated at each step of each game. The following two remarks describe these strategies.\par
For the game with perfect local information we consider the following strategies:

\begin{remark}
\begin{enumerate}
  \item The strategy profile $u(k)=\mathbb{0}_{\sum_{i=1}^{N}d_i}$, $k=1...K$.
  \item The strategy profile $u^*(k)$, $k=1...K$:
\begin{equation*}
u^i_j(k)^*= \left\{
\begin{array}{ll}
      1, \quad \textrm{if}\quad x_i(k)=1 \quad \textrm{or} \quad\{ x_i(k)=0 \quad \textrm{and} \quad (r_j(k)=0 \quad \textrm{or} \quad r_j(k)=0)\}\\
      \textrm{The solution of algorithm \ref{algorithm1}}, \quad \textrm{if}\quad x_i(k)=0 \quad \textrm{and} \quad 1\leq r_j(k)<R\\
\end{array}
\right.
\end{equation*}
In the execution of algorithm \ref{algorithm1} in this equation, the set $ \textrm{Infected}_i$ is defined as follows:
\begin{equation*}
  \textrm{Infected}_i=\{j \in N_i: 1\leq r_j(k)<R\}
\end{equation*}
  \item Any step-wise interchange of the two previous strategy profiles.
\end{enumerate}
\end{remark}

and for the game with statistical information we consider the following strategies:
\begin{remark}
\begin{enumerate}
  \item The strategy profile $u(k)=\mathbb{0}_{\sum_{i=1}^{N}d_i}$, $k=1...K$.
  \item The strategy profile $u^*(k)$, $k=1...K$: The solution of Algorithm \ref{algorithm2}, where $p_x^0=p_x(k)$ follows the rule \eqref{px_k}.
  \item Any stepwise interchange of the two previous strategy profiles.
\end{enumerate}
\end{remark}

In practice, since both algorithms are initialized with all the social contacts being active, they converge in the second category of strategies for both cases.\par
The simulations presented here have the following parameters. The underlying graph topology is a random graph \cite{random} with $N=1000$ agents and adjacency probability $p=1\%$, and thus average degree $\bar{d}=10$. The recovery period is 14 days. The sociability parameters $s_{ij}$ are random numbers in $(0,1)$. The agents are divided into two groups the vulnerable and the non-vulnerable. For the vulnerable $G_i=10000$ and for the non-vulnerable $G_i=1000$. The percentage of the vulnerable in the community is $20\%$. The initial percentage of infected agents is $1\%$. The basic reproduction number of the disease $R_0=2$, so as for the disease to be epidemic and for social distancing to be necessary.\par
In Figure \ref{fig:distr_info_3_curves} we indicate the effects of the social distancing games with statistical information and with perfect local information on the disease prevalence and on the sociability of the agents and compare these two games. In Table \ref{table:4} we present some numerical characteristics of these curves.

\begin{figure}[h!]
\centering
  \includegraphics[width=\linewidth,height=8cm]{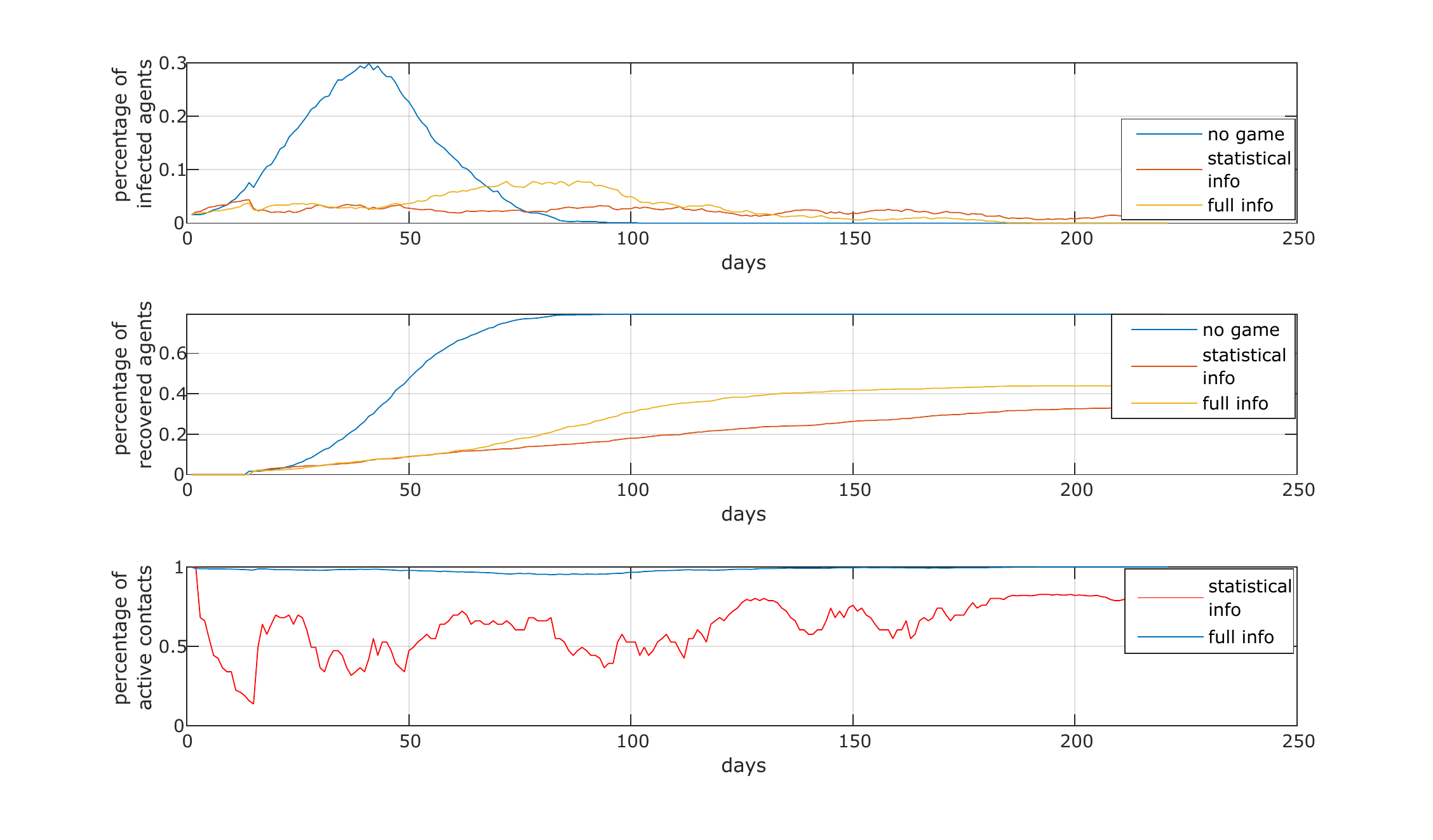}
  \caption{Infection, recovery and sociability curves for the game with information for the distribution of the states}
  \label{fig:distr_info_3_curves}
\end{figure}

\begin{table}[h!]
\centering
\begin{tabular}{ |p{5cm}||p{3cm}|p{3cm}|p{4cm}|  }
 \hline
 \multicolumn{4}{|c|}{Characteristics of Infection and Sociability for the Games with Different Information} \\
 \hline
   &  Infection Peak & Total Infection & Minimum Sociability \\
 \hline
 No game  & 30.3\%  & 78.7\% & 100\%\\
 Perfect local feedback   & 9.1\% & 43.7\% &  94.2\% \\
 Information for the distribution of the infected   & 4.2\% & 33.5\% & 22.9\%\\
 \hline
\end{tabular}
\caption{Table to compare the infection and sociability for the games with different information}
\label{table:4}
\end{table}

We conclude that the game with information for the distribution of the states almost extinguish the epidemic behavior of the disease - the infection curve is similar with the infection curve of a disease with $R_0 \leq 1$ - while the perfect local state feedback game reduces significantly the disease outspread, but it does not extinguish the epidemic. However, this result comes at a high cost for the agents, since in presence of an augmented fear of infection due to lack of information about the health status of their friends, they avoid many of their social interactions. In comparison with the perfect local state feedback repeated game the social distancing is large, there exist only $22.9\%$ active social interactions, than $94.2\%$ in the other case. This remarkable difference on the agents behavior affects significantly their payoffs. As we observe in Figure \ref{fig:payoffs_comparison}, the average payoff of the game with perfect local information are much higher than the average payoff of the game with statistical information. Moreover, we must underline that for the vulnerable agents the difference of their payoffs between the two games is large, as they pay much higher the cost of not being informed about the health state of their contacts and get infected. \par

\begin{figure}[h!]
\centering
  \includegraphics[width=\linewidth,height=8cm]{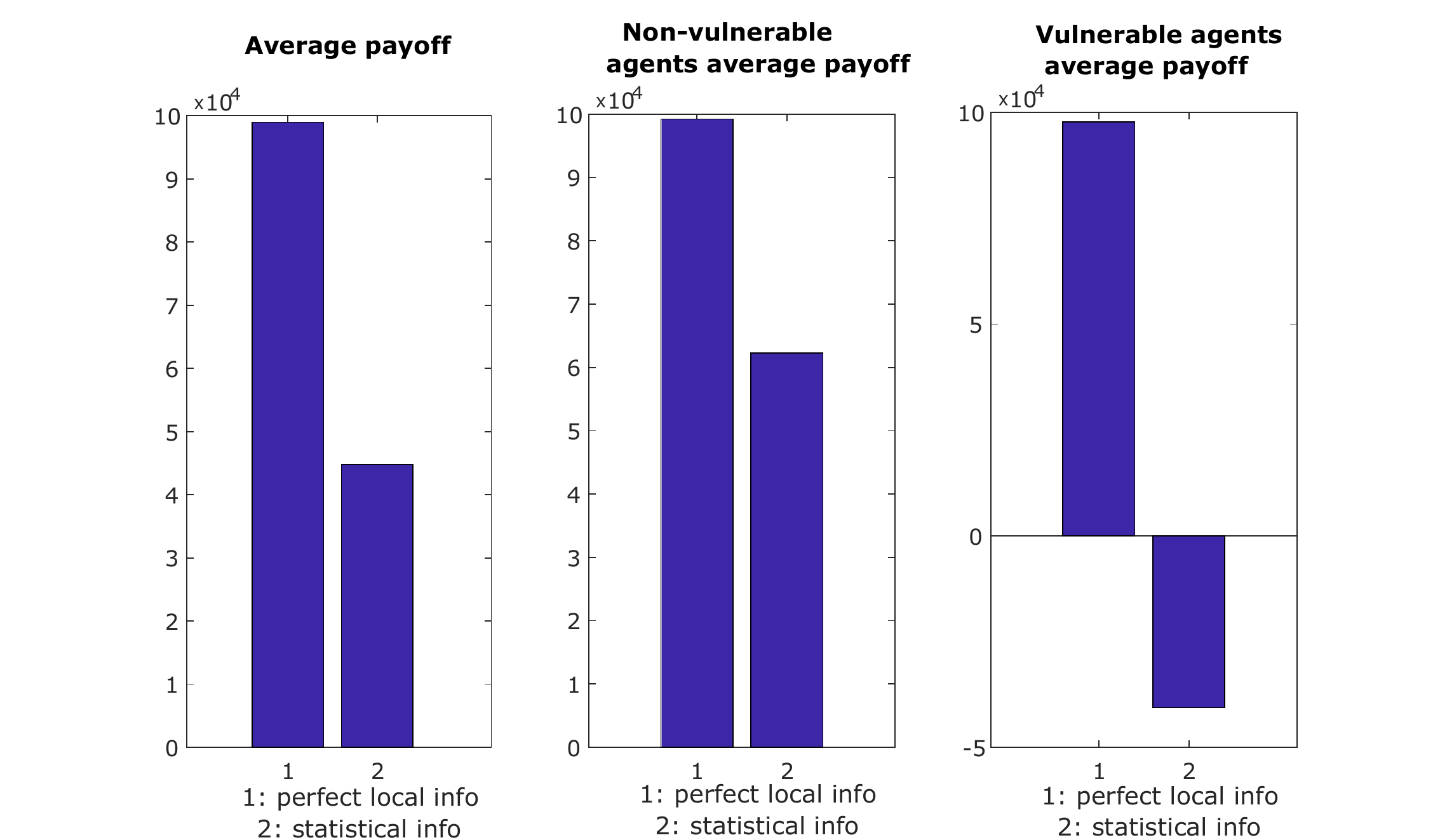}
  \caption{Comparison of the payoffs of the agents for the two games}
  \label{fig:payoffs_comparison}
\end{figure}

Due to that fact the vulnerable agents can be considered as key players for these games, since they tend to play conservatively and thus enhance the social distancing. So, in Figure \ref{fig:vulnerables_full_info} we show the effect of the percentage of vulnerable agents in the community to the infection peak and to the total number of infected agents for the game with perfect local information and in Figure \ref{fig:vulnerables_distr_info} we show the same effects for the game with statistical information. In these figures the mean value of $30$ simulations is depicted at each point of the plots and the red lines are the linear regression curves for our experiments on different percentages.

\begin{figure}[h!]
\centering
  \includegraphics[width=\linewidth,height=8cm]{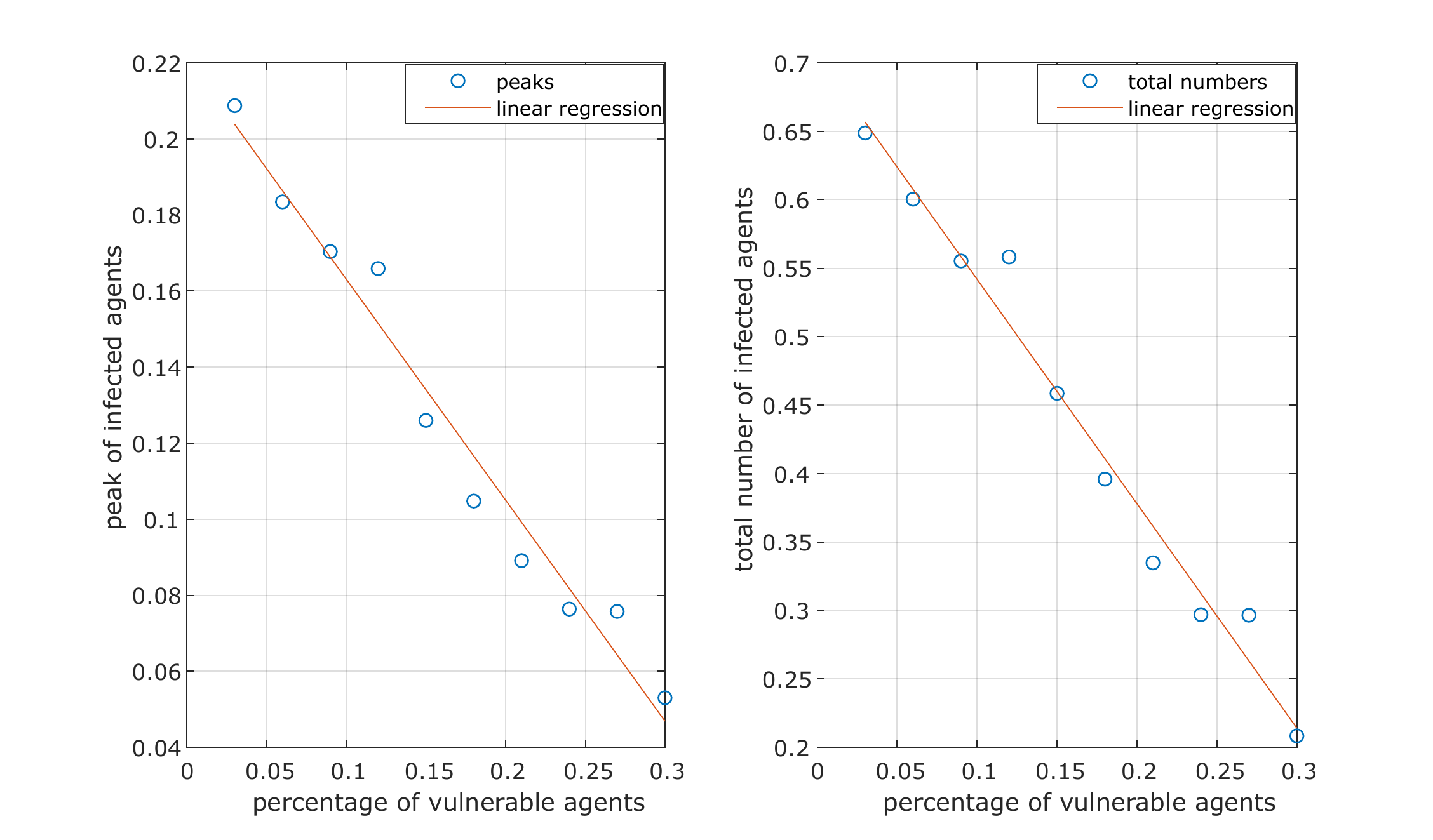}
  \caption{Correlation of the percentage of vulnerable agents with the infection outspread for the perfect local feedback information game}
  \label{fig:vulnerables_full_info}
\end{figure}

\begin{figure}[h!]
\centering
  \includegraphics[width=\linewidth,height=8cm]{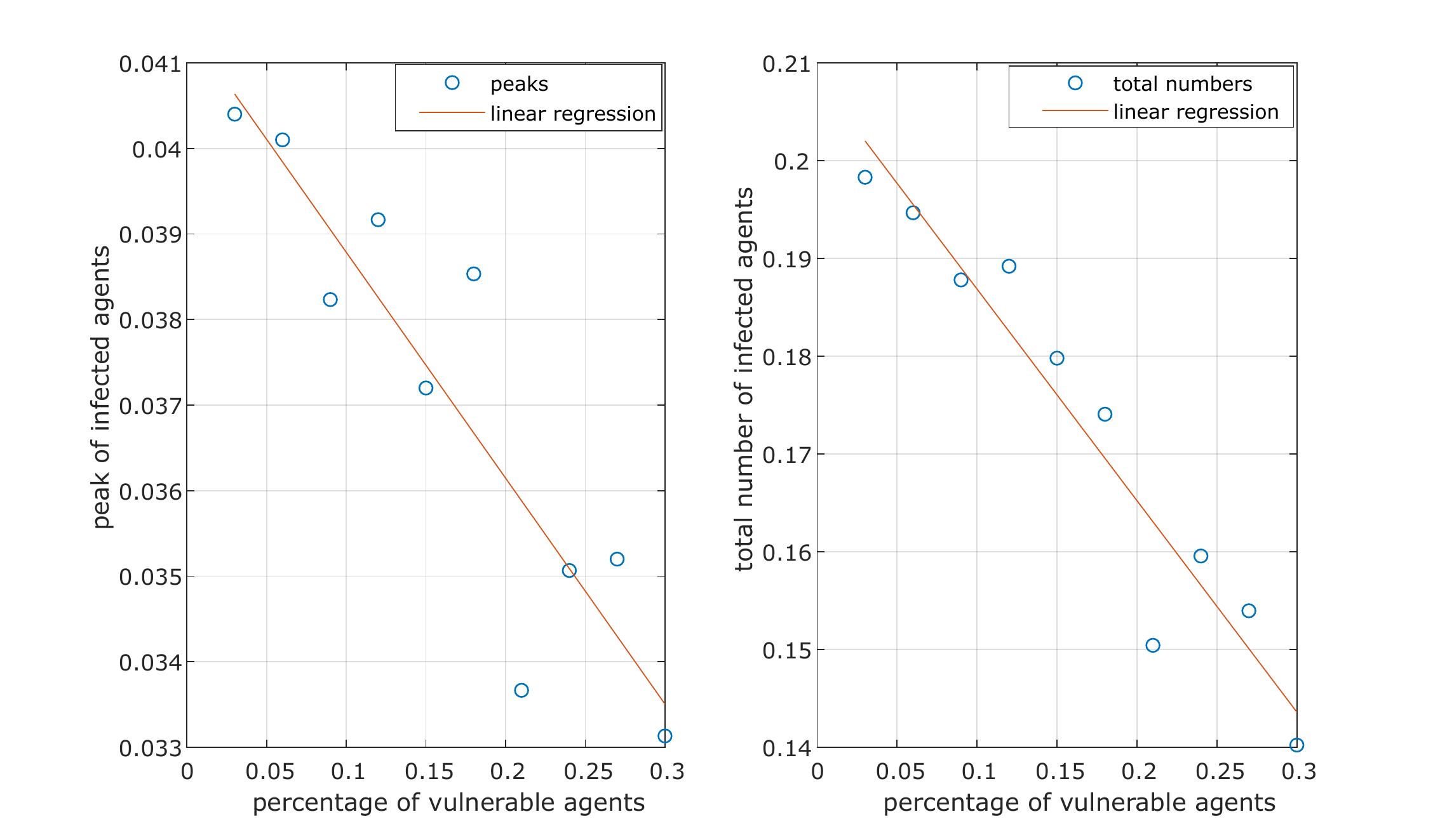}
  \caption{Correlation of the percentage of vulnerable agents with the infection outspread for the game with information for the distribution of the states}
  \label{fig:vulnerables_distr_info}
\end{figure}

The result is the expected one, that in both games the percentage of the vulnerable agents is negatively correlated to the disease outspread. However, in the game with statistical information (Figure \ref{fig:vulnerables_distr_info}) their role seem to be less important than in the perfect local information game, because the variations of the peaks and the total number of infected agents with respect to the vulnerable agents percentage are by far smaller.

\section{Case Studies and Discussion}\label{discussion}

In the previous sections, we have analyzed and compared the Nash equilibrium strategies of the agents for the cases of perfect local state feedback information and statistical information for the distribution of the states. In this section, we consider several variations of the initial problem and examine, through simulations, the effects of the varying parameters on the behavior of the agents and on the outspread of the epidemic. The first variation  we study concerns the quality of the available information, for the case that the agents possess  statistical information for the distribution of the states. The second variation considers the risk perception, modeled by the vulnerability parameters, to depend on the infection outspread and the capacity of health care system. And the third variation takes under consideration the effects of the graph topology.

\subsection{Fake information for the distribution of the states}
A first modified scenario we examine is the case that the information the agents possess about the distribution of the states is fake or biased. This is an interesting and in some cases realistic scenario, since the agents are rarely able or have the time to investigate verified data about the outspread of the disease, but they usually get informed through mass media or social media. Consequently, the information they get is usually exaggerated or understated. The spread of fake news is another factor affecting the information quality and thus the decisions of the agents. Moreover, in many cases the lack of diagnostic tests in the community makes the knowledge of the accurate infection level impossible.
\vskip 0.3cm
So, we consider the following modification of the model of section \ref{s.distr_info}:
\begin{equation}\label{fake_info}
  p_x^f=f p_x
\end{equation}
where $p_x^f$ is the available fake information of the agents and $f$ is a coefficient indicating its deviation from the actual information $p_x$. So, we get the following simulations (Figure:\ref{fig:fake_info_3_curves}) indicating the effects of an overestimation of the infection level ($f=2$) and an underestimation of the infection level ($f=0.5$), in comparison with the game with actual information (all the parameters are the same with the other simulations, as in section \ref{simulation_parameters}).
\begin{figure}[h!]
\centering
  \includegraphics[width=\linewidth,height=8cm]{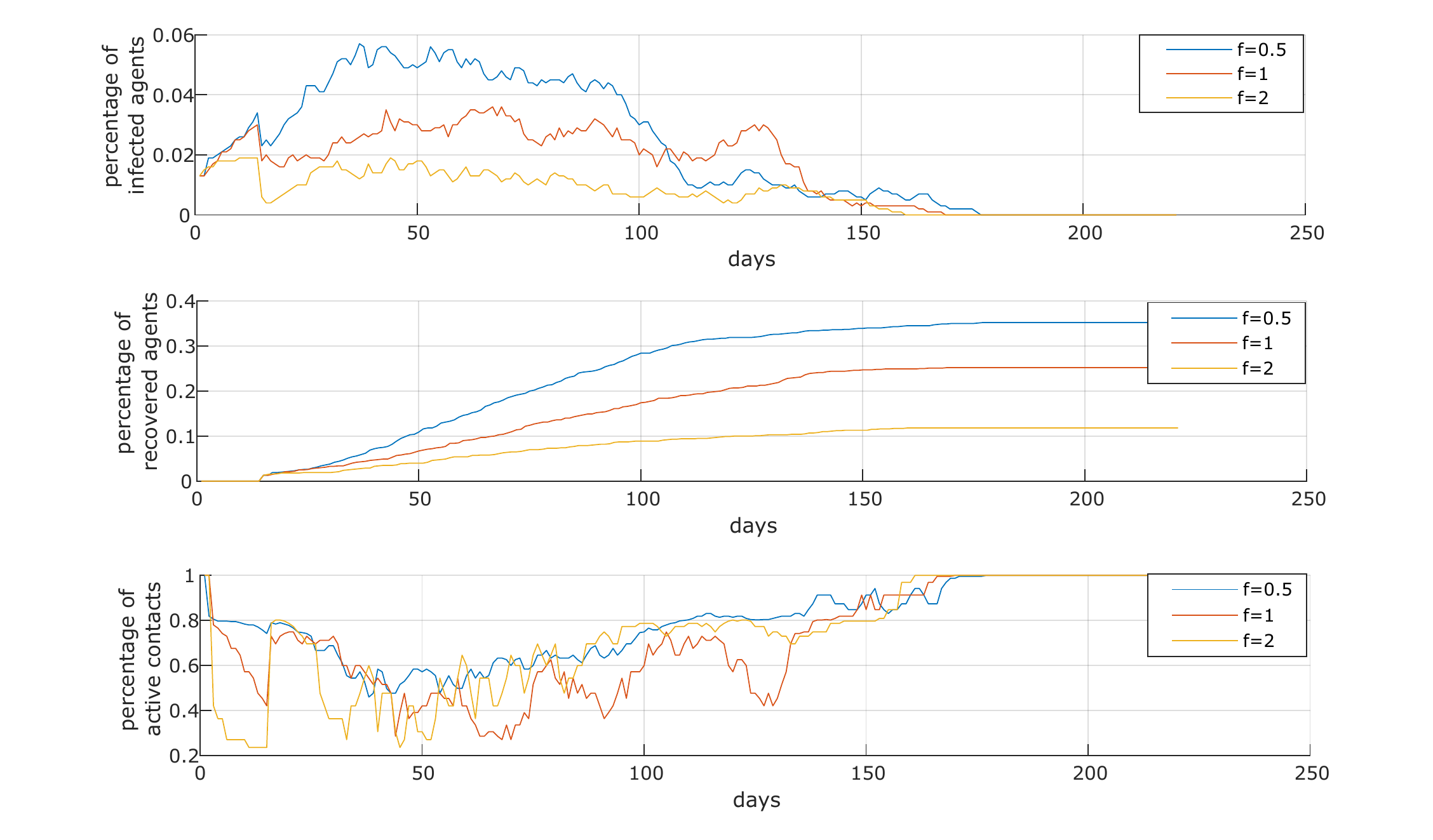}
  \caption{Infection, recovery and sociability curves for games with fake information for the distribution of the states }
  \label{fig:fake_info_3_curves}
\end{figure}
In table \ref{table:5} are presented some numerical characteristics of the previous curves:\par
\begin{table}[h!]
\centering
\begin{tabular}{ |p{6.5cm}||p{2.8cm}|p{2.8cm}|p{3.8cm}|  }
 \hline
 \multicolumn{4}{|c|}{Fake Information Consequences} \\
 \hline
   &  Infection Peak & Total Infection & Minimum Sociability\\
 \hline
 Actual Information  & 3.7\% & 25.6\% & 23.9\%\\
 Overestimation of Infection $(2p_x)$ & 1.9\% & 11.1\% & 19.5\%\\
 Underestimation of Infection $(0.5p_x)$ & 5.9\% & 36.5\% & 52.6\%\\
 \hline
\end{tabular}
\caption{Table to compare the infection for games with actual and fake information}
\label{table:5}
\end{table}
We observe that in the case of an overestimation of the infection level the agents care more to follow social distancing and the disease prevalence is kept at low levels, while in the case of underestimation of the infection the agents do not care so much and the disease prevalence is higher. It seems rational that the agents are more benefited from a low prevalence so it may be profitable for them to receive an overestimation of the infection level. However, when applying social distancing they pay the costs of the effort to eradicate the epidemic, so it is interesting to examine how their payoffs vary when they receive fake or biased information.
\begin{figure}[h!]
\centering
  \includegraphics[width=\linewidth,height=8cm]{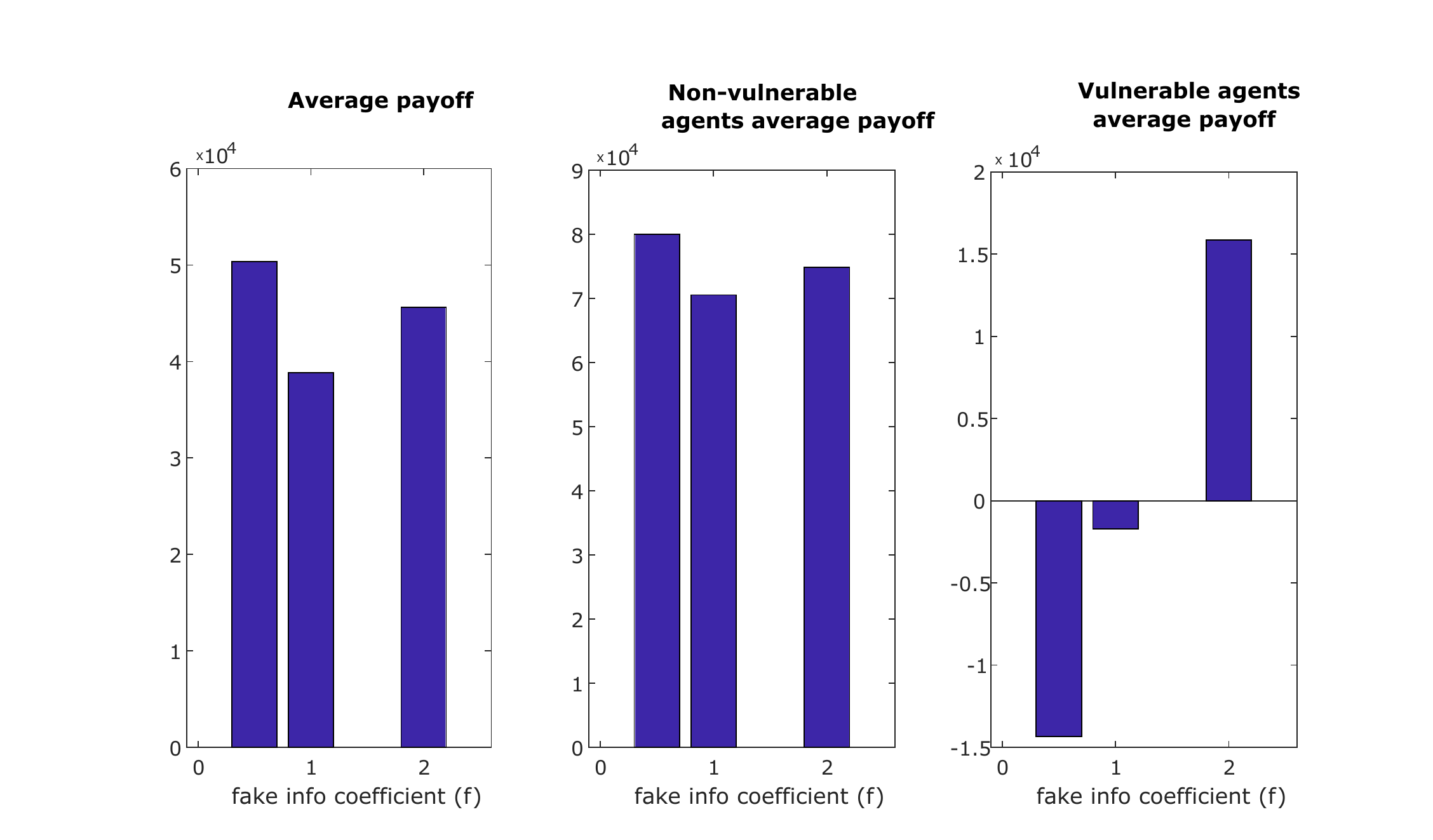}
  \caption{Payoffs of the agents in the case of fake information}
  \label{fig:fake_payoffs}
\end{figure}
From Figure \ref{fig:fake_payoffs}, we observe that the non vulnerable agents are slightly benefited from an underestimation of the infection level, since they are not precarious and their profits depend mostly on their social contacts. However, they have also small gains in the case of an overestimation of the infection level since they stay safe and not infected. Contrary to the non vulnerable agents, the vulnerable agents are damaged from an underestimation of the infection level, since they may get infected and suffer a lot and they are benefited significantly from an overestimation of the infection level, which scares the whole society and leads to strict social distancing, saving them this way from a possible infection.\par
It is also interesting to point out the correlation of the infection outspread and the agents' payoffs with the fake information coefficient $(f)$. As we can see in Figure \ref{fig:fake_info_peaks_sums} the infection outspread and the fake information coefficient have a definitely negative correlation, with overestimation leading to very low infection levels and underestimation to high infection levels. Moreover, interestingly, the agents' payoffs present a minimum average value when the information is near to the actual one and they seem to be benefited from fake or biased information. However, we should point out here that the average payoff of all the agents is illustrated in Figure \ref{fig:fake_info_peaks_sums} and this is the reason for the existence of that minimum, since, according to Figure \ref{fig:fake_payoffs}, the non vulnerable agents - who are the majority - are slightly benefited from a small $f$ and the vulnerable agents - even if they are minority - are benefited significantly form a larger $f$, so the average payoff has this behavior.

\begin{figure}[h!]
\centering
  \includegraphics[width=\linewidth,height=8cm]{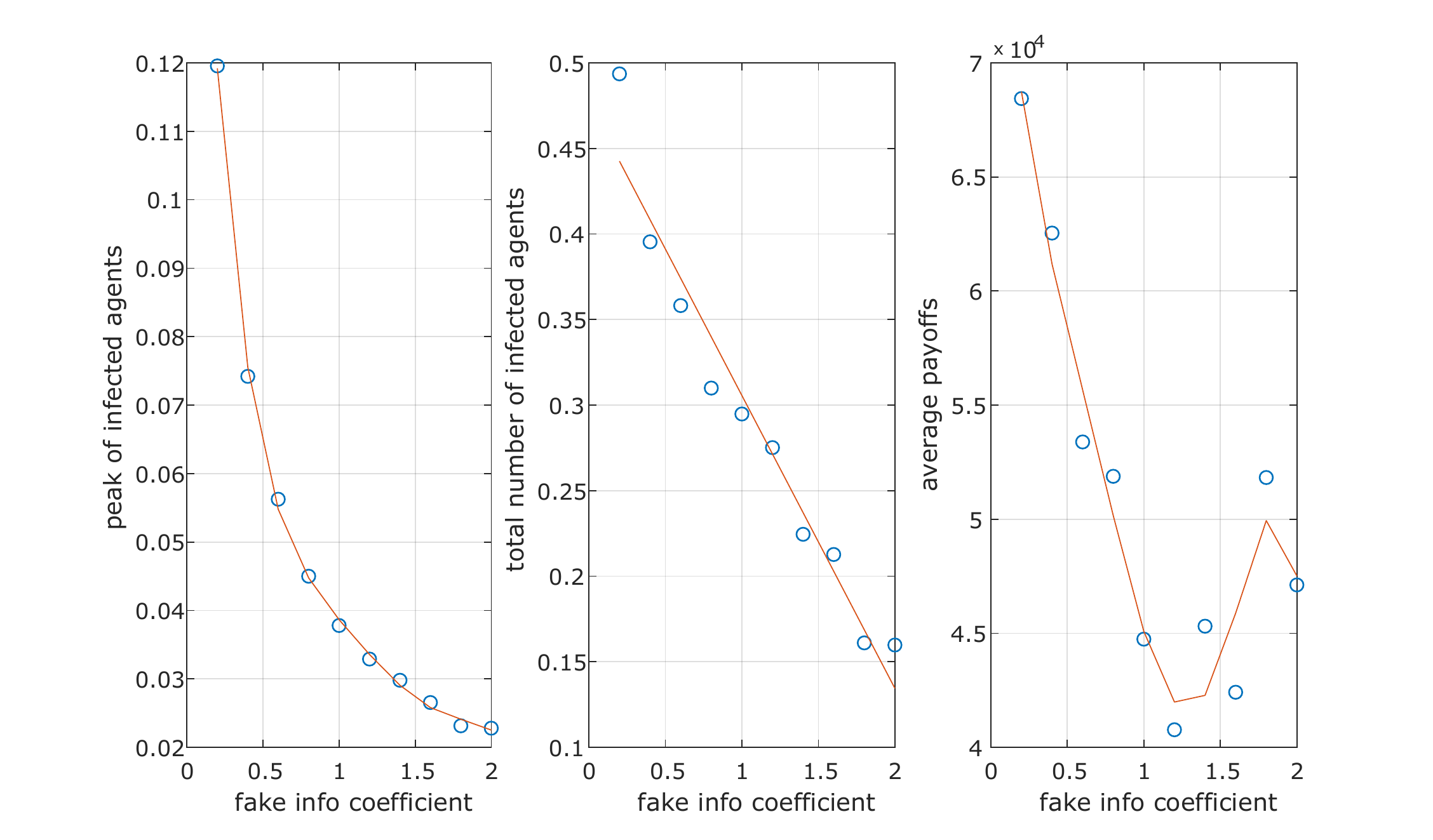}
  \caption{Correlation of the coefficient of fake information with the infection outspread and the agents' payoffs}
  \label{fig:fake_info_peaks_sums}
\end{figure}

These observations could be useful for a social planer, with an aim to avoid the spreading of the disease, in order to achieve social distancing without imposing it by law but by manipulating the agents strategies through the broadcasted fake or biased information.

\subsection{Vulnerabilities depending on the infection outspread and the health care system capacity}
Another scenario is that the vulnerability parameters of the agents ($G_i$) depend on the level of infection in the community. This is an interesting scenario in practice, since the health systems worldwide have finite (and usually small) capacity, so if the number of infected agents who need health care pass a certain level it is not probable for the next agents who will get infected to have access in the necessary facilities. \par
We model this phenomenon considering the vulnerability parameters as functions of the infection ratio, in the model of section \ref{s.distr_info} . At first, we examine the case of linear dependence:

\begin{equation}\label{G_px_prop}
  G_i=G_i(p_x)=G_i^0\alpha p_x
\end{equation}

$G_i^0$ are the constant vulnerability parameters used in all the previous simulations. Choosing $\alpha=\frac{1}{p_x^{\text{ref}}}$ we can define a reference infection level $p_x^{\text{ref}}$, where the agents will play as in the case of constant vulnerability parameters $G_i^0$. Below this level, they will be indifferent for the effects of the disease on them and care more for their social interactions and above this level they will be more worried about the disease and follow social distancing strategies.\par
This is confirmed by Figure: \ref{fig:G_px_proportional}, where all the parameters, except the vulnerability parameters, are the same with the other simulations, as in section \ref{simulation_parameters}.

\begin{figure}[h!]
\centering
  \includegraphics[width=\linewidth,height=8cm]{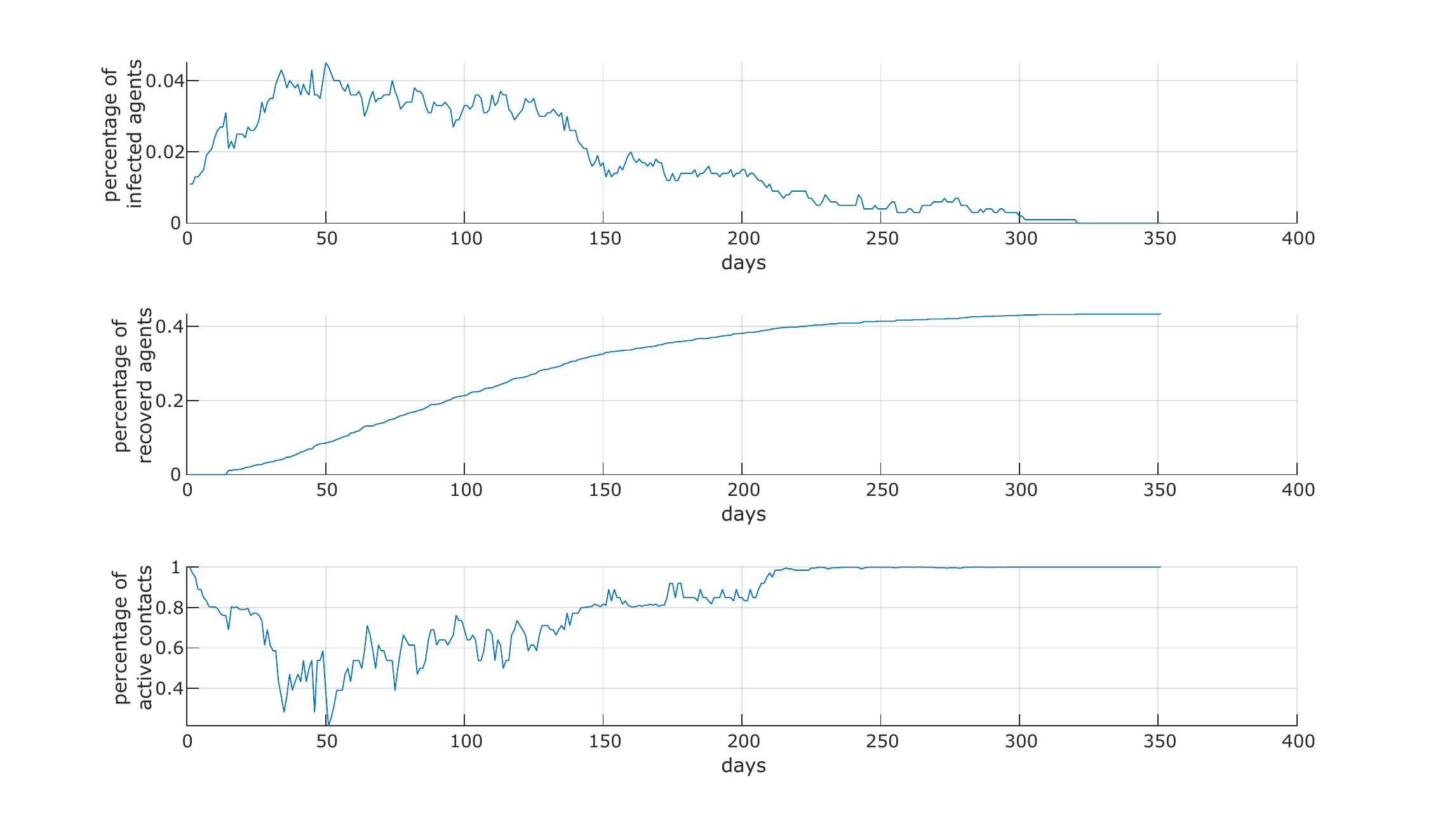}
  \caption{Infection, recovery and sociability curves when the vulnerability parameters have a proportional dependence on the infection outspread}
  \label{fig:G_px_proportional}
\end{figure}
It is interesting to point that the agents do keep the infection level below the reference value, in this simulation $p_x^{\text{ref}}=5\%$, for all the time.\par
We also examine the case of step-function dependence, where there exists a critical infection level $p_x^{\text{cr}}$ above which the agents play with parameters $G_i^0$ and below which they care very little about the effects of the disease on them.

\begin{equation}
G_i=G_i(p_x)= \left\{
\begin{array}{ll}
      G_i^0/M &, p_x<p_x^{\text{cr}} \\
      G_i^0 &, else\\
\end{array}
\right.
\end{equation}
where $M$ is a large number e.g., $M=100$. In Figure \ref{fig:G_px_step}, $p_x^{\text{cr}}=5\%$. \par
\begin{figure}[h!]
\centering
  \includegraphics[width=\linewidth,height=8cm]{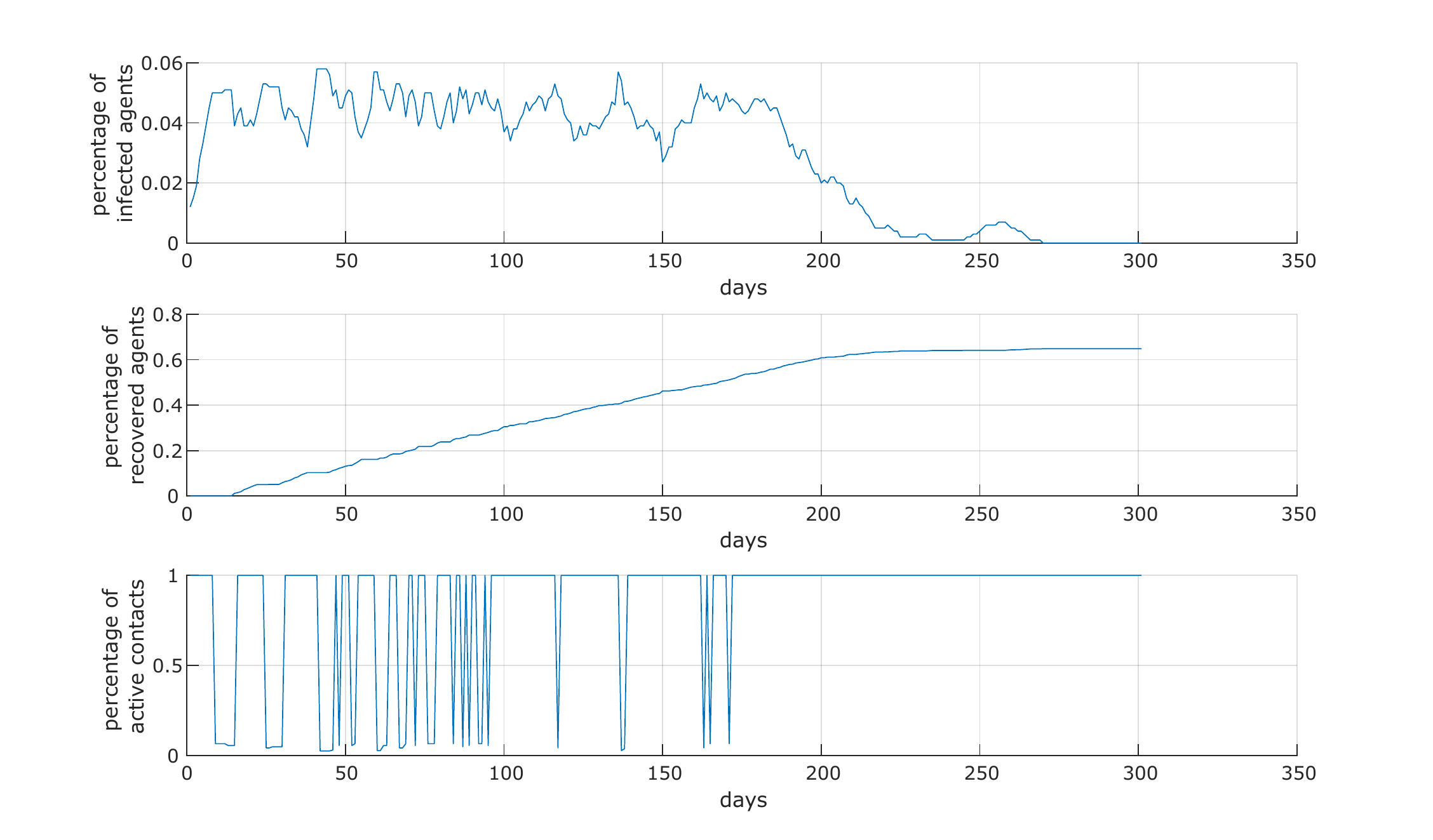}
  \caption{Infection, recovery and sociability curves when the vulnerability parameters have a step-wise dependence on the infection outspread}
  \label{fig:G_px_step}
\end{figure}
It should be pointed here that in the first case (Fig:\ref{fig:G_px_proportional}) the peak of the infection is $4.3\%$ and the total number of infected agents is $41.3\%$, while in the second case (Fig:\ref{fig:G_px_step}) the peak is $5.1\%$ and the total number of infected agents is $52.4\%$, which indicates a much worse behavior of the agents in the second case.

\subsection{ Effects of the graph topology on the outspread of the disease}
Except of the agents strategies of social distancing, another very important factor which affects the outspread of the disease is the topology of the underlying network which represents the social interactions of the agents. So, we present here some simulations to indicate the effects of the graph topology on the spreading of the disease and  the effectiveness of the agents strategies.\par
At first, we do not consider a social distancing game, so every two neighbors communicate freely with each other ($w_{ij}=a_{ij}$). In Figure \ref{fig:topology_comparison} we observe the infection curves for four different graph topologies: random graph \cite{random}, stochastic block model, scale free network \cite{Barabasi},\cite{Scale_free} and small world network \cite{small_world}. In every case we have chosen the network parameters in a way that the graphs have almost the same average degree ($\bar{d}\approx 10$), since it is a key parameter determining the scale of the infection level, as shown in  \eqref{contamination_probability}. In table \ref{table:2} we present some numerical characteristics of these curves.

\begin{figure}[h!]
\centering
  \includegraphics[width=\linewidth,height=8cm]{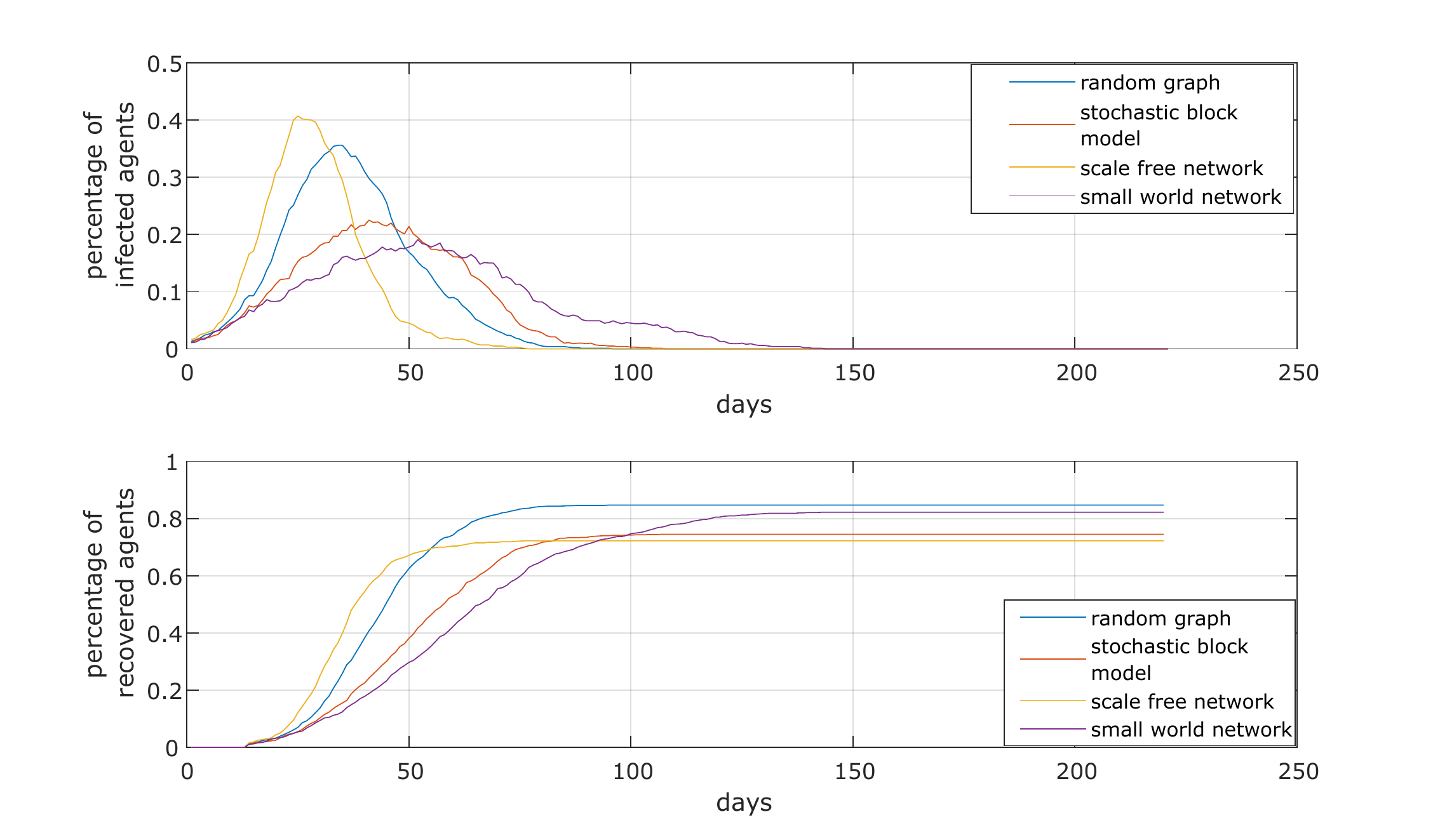}
  \caption{Infection curves for different graph topologies for the case of no social distacning}
  \label{fig:topology_comparison}
\end{figure}

\begin{table}[h!]
\centering
\begin{tabular}{ |p{5cm}||p{3cm}|p{3cm}|p{3cm}|  }
 \hline
 \multicolumn{4}{|c|}{Topology Comparison (No Social Distancing)} \\
 \hline
   & Average Degree & Infection Peak & Total Infection\\
 \hline
 Random Graph   & 9.8   & 36.2\% & 82.8\% \\
 Stochastic Block Model &   10.2  & 22.1\% & 74.1\% \\
 Scale Free Network & 9.9 & 41.3\% & 72.3\% \\
 Small World Network   & 10 & 17.9\% & 80.9\% \\
 \hline
\end{tabular}
\caption{Table to compare the infection for different graph topologies in the case of no social distancing}
\label{table:2}
\end{table}

We can make several interesting observations from Figure \ref{fig:topology_comparison} and Table \ref{table:2}. The scale free network with some central nodes with large degrees presents very early a sharp and high infection, while the small world network, which arises from a lattice and thus has high regularity and all the nodes have almost the same degrees, present a much lower but extended infection curve. However, the total numbers of infected agents are almost the same. Moreover, the stochastic block model - being an ill connected coalition of well connected random graphs - has a similar, in shape, reaction curve with the random graph, but with lower peak and significantly lower total number of infected agents.\par
Next, we proceed with the numerical study of the social distancing games with the two different information structures. We begin with the game with perfect local information, illustrated in Figure \ref{fig:topology_comparison_full_info} and Table \ref{table:3}.

\begin{figure}[h!]
\centering
  \includegraphics[width=\linewidth,height=8cm]{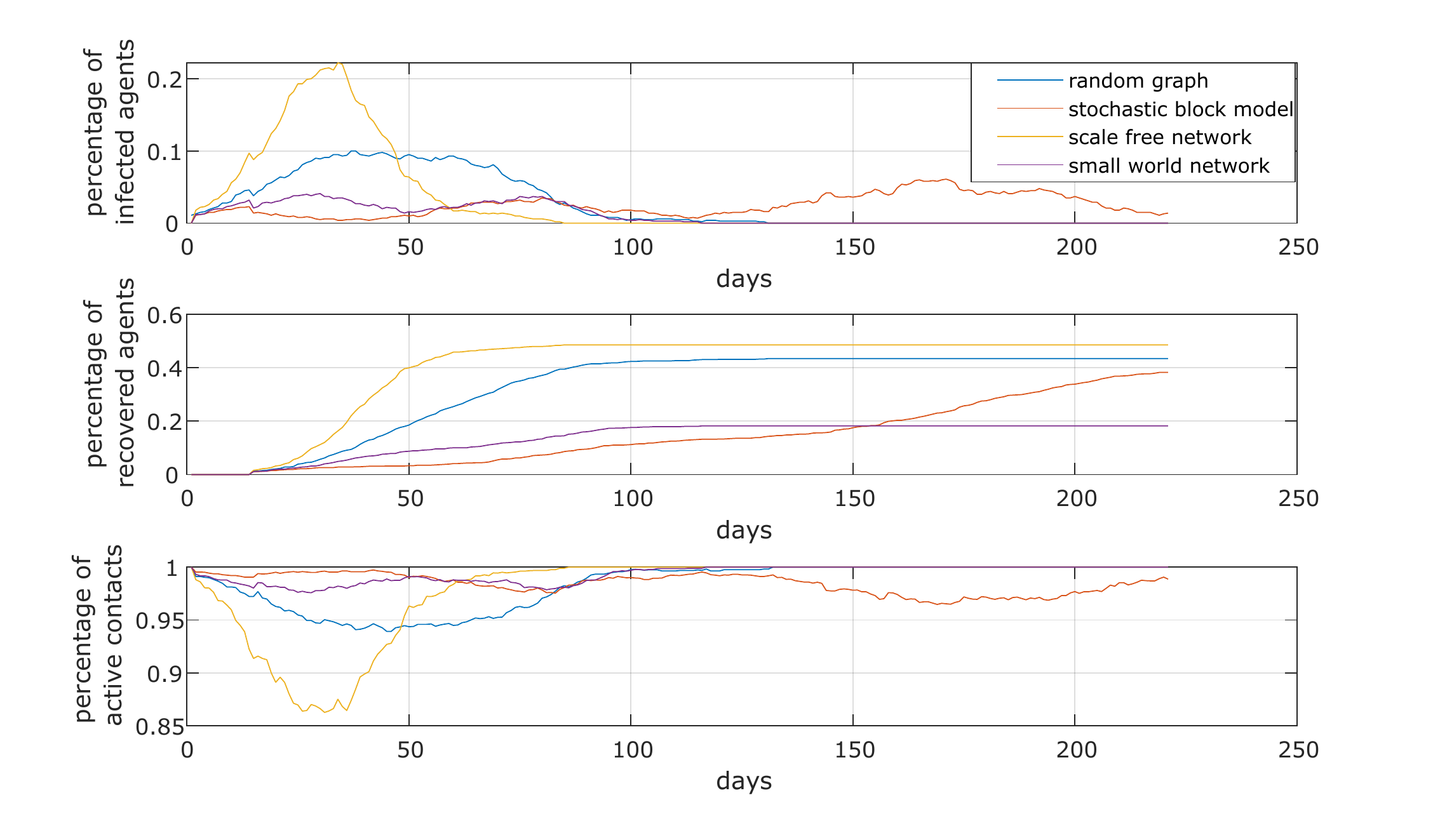}
  \caption{Infection curves for different graph topologies for the game with perfect local information}
  \label{fig:topology_comparison_full_info}
\end{figure}

\begin{table}[h!]
\centering
\begin{tabular}{ |p{5cm}||p{3cm}|p{3cm}|p{4cm}|  }
 \hline
 \multicolumn{4}{|c|}{Topology Comparison (Perfect Info Game)} \\
 \hline
   & Infection Peak & Total Infection & Minimum Sociability\\
 \hline
 Random Graph   & 10.2\% & 42.2\% & 94.8\% \\
 Stochastic Block Model  & 6.1\% & 38.8\% & 96.3\%\\
 Scale Free Network & 21.7\% & 49.8\% & 85.7\%\\
 Small World Network  & 4.1\% & 18.9\% & 97.5\%\\
 \hline
\end{tabular}
\caption{Table to compare the infection for different graph topologies for the case of the perfect local information game}
\label{table:3}
\end{table}

From these simulations we can derive the following conclusions. Firstly, in the scale free network the disease spreads quickly, even in the case that the agents follow social distancing strategies. Thus, we observe the highest peak, the greatest total prevalence and consequently the lowest level on the agents sociability in their effort to flatten this curve. In the cases of small world network and stochastic block model the infection peaks are low. Thus, the agents are not so concerned about the disease and relax their social distancing, resulting in second waves of the epidemic in both cases. However, in the case of the stochastic block model the duration of the waves is larger and consequently the total prevalence is also larger. Finally, in the case of the random graph we observe a greater peak than in the stochastic block model, but due to the response of the agents and the absence of remote cliques - which may act as sources of infection - the disease is eliminated in this case and we do not observe a second wave, resulting in a similar total prevalence with the stochastic block model.\par
Following that, we present simulations for the game with statistical information, illustrated in Figure \ref{fig:topology_comparison_distr_info} and Table \ref{table:4}.

\begin{figure}[h!]
\centering
  \includegraphics[width=\linewidth,height=8cm]{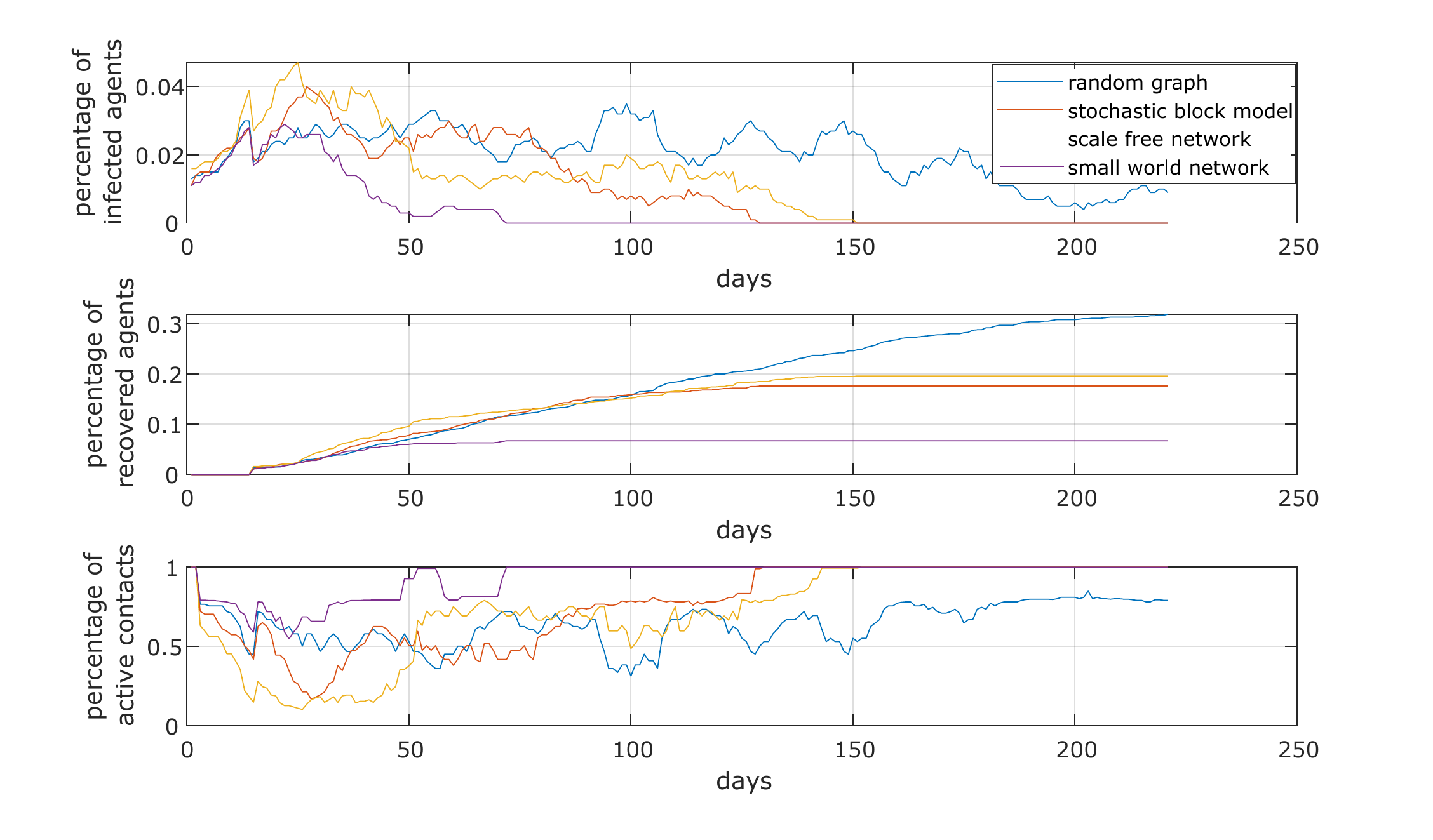}
  \caption{Infection curves for different graph topologies for the game with statistical information}
  \label{fig:topology_comparison_distr_info}
\end{figure}

\begin{table}[h!]
\centering
\begin{tabular}{ |p{5cm}||p{3cm}|p{3cm}|p{4cm}|  }
 \hline
 \multicolumn{4}{|c|}{Topology Comparison (Statistical Info Game)} \\
 \hline
   & Infection Peak & Total Infection & Minimum Sociability\\
 \hline
 Random Graph    & 3.5\% & 31.2\% & 36.1\%\\
 Stochastic Block Model & 4.0\% & 18.3\% & 23.4\%\\
 Scale Free Network & 4.6\% & 19.9\% & 17.2\%\\
 Small World Network  & 2.9\% & 7.1\% & 52.3\%\\
 \hline
\end{tabular}
\caption{Table to compare the infection for different graph topologies for the case of the game with statistical information}
\label{table:4}
\end{table}

There are several interesting observations for this case also. At first, in all graph topologies the infection level is kept very low but with a high cost on the sociability of the agents. Secondly, there do not exist significant differences on the peaks of the infection but there exist on the epidemic's duration, affecting this way its total prevalence. In the small world network the epidemic is eliminated quickly and with a comparatively small effort from the agents, resulting in a low total prevalence. In the scale free network and the stochastic block model the epidemic is also eliminated but has a relatively larger duration, resulting in similarly larger prevalence. Finally, in the random graph topology the epidemic lasts long and has the greatest prevalence. This phenomenon is probably a result of the well mixing of the agents, that arises more in the random graph topology, which contributes to the persistence of the disease even in the case that the agents cut several of their contacts.

\section{Conclusion}
A game-theoretic approach of social distancing has been considered. In the simple model proposed, the main parameters under examination are the network describing the structure of the interactions among the agents, which changes according to their rationally chosen strategies and the available information during the decision making. The effects of the spontaneous social distancing behavior on the prevalence of the epidemic is investigated both analytically and numerically through simulations on artificial networks. At the current stage, the proposed model is not intended for quantitative policy suggestions, since on the one hand it is simplistic and on the other hand the knowledge of realistic values for the parameters modeling human behavior requires real observations, many data and a proper statistical processing of them. However, it may be useful to offer a paradigm on the way the agents decide to adopt social distancing and the effects of these decisions on the prevalence of an epidemic.

\bibliographystyle{IEEEtran}
\bibliography{refs2}

\begin{thebibliography}{10}
\providecommand{\url}[1]{#1}
\csname url@samestyle\endcsname
\providecommand{\newblock}{\relax}
\providecommand{\bibinfo}[2]{#2}
\providecommand{\BIBentrySTDinterwordspacing}{\spaceskip=0pt\relax}
\providecommand{\BIBentryALTinterwordstretchfactor}{4}
\providecommand{\BIBentryALTinterwordspacing}{\spaceskip=\fontdimen2\font plus
\BIBentryALTinterwordstretchfactor\fontdimen3\font minus
  \fontdimen4\font\relax}
\providecommand{\BIBforeignlanguage}[2]{{%
\expandafter\ifx\csname l@#1\endcsname\relax
\typeout{** WARNING: IEEEtran.bst: No hyphenation pattern has been}%
\typeout{** loaded for the language `#1'. Using the pattern for}%
\typeout{** the default language instead.}%
\else
\language=\csname l@#1\endcsname
\fi
#2}}
\providecommand{\BIBdecl}{\relax}
\BIBdecl

\bibitem{Bernoulli}
D.~Bernoulli, ``Essai d'une nouvelle analyse de la mortalit{\'e} caus{\'e}e par
  la petite v{\'e}role, et des avantages de l'inoculation pour la
  pr{\'e}venir,'' \emph{Histoire de l'Acad., Roy. Sci.(Paris) avec Mem}, pp.
  1--45, 1760.

\bibitem{Kermack}
W.~O. Kermack and A.~G. McKendrick, ``A contribution to the mathematical theory
  of epidemics,'' \emph{Proceedings of the royal society of london. Series A,
  Containing papers of a mathematical and physical character}, vol. 115, no.
  772, pp. 700--721, 1927.

\bibitem{Ross}
R.~Ross, ``An application of the theory of probabilities to the study of a
  priori pathometry,'' \emph{Proceedings of the Royal Society of London. Series
  A, Containing papers of a mathematical and physical character}, vol.~92, no.
  638, pp. 204--230, 1916.

\bibitem{Pastor-Satorras}
R.~Pastor-Satorras, C.~Castellano, P.~Van~Mieghem, and A.~Vespignani,
  ``Epidemic processes in complex networks,'' \emph{Reviews of modern physics},
  vol.~87, no.~3, p. 925, 2015.

\bibitem{Newman1}
M.~E. Newman, ``Spread of epidemic disease on networks,'' \emph{Physical review
  E}, vol.~66, no.~1, p. 016128, 2002.

\bibitem{Newman2}
C.~Moore and M.~E. Newman, ``Epidemics and percolation in small-world
  networks,'' \emph{Physical Review E}, vol.~61, no.~5, p. 5678, 2000.

\bibitem{Meyers}
L.~A. Meyers, M.~Newman, and B.~Pourbohloul, ``Predicting epidemics on directed
  contact networks,'' \emph{Journal of theoretical biology}, vol. 240, no.~3,
  pp. 400--418, 2006.

\bibitem{Garnett}
G.~P. Garnett and R.~M. Anderson, ``Sexually transmitted diseases and sexual
  behavior: insights from mathematical models,'' \emph{Journal of Infectious
  Diseases}, vol. 174, no. Supplement\_2, pp. S150--S161, 1996.

\bibitem{Sander}
L.~Sander, C.~Warren, I.~Sokolov, C.~Simon, and J.~Koopman, ``Percolation on
  heterogeneous networks as a model for epidemics,'' \emph{Mathematical
  biosciences}, vol. 180, no. 1-2, pp. 293--305, 2002.

\bibitem{Epstein1}
J.~M. Epstein, ``Modelling to contain pandemics,'' \emph{Nature}, vol. 460, no.
  7256, pp. 687--687, 2009.

\bibitem{Epstein2}
J.~M. Epstein, J.~Parker, D.~Cummings, and R.~A. Hammond, ``Coupled contagion
  dynamics of fear and disease: mathematical and computational explorations,''
  \emph{PLoS One}, vol.~3, no.~12, 2008.

\bibitem{Cliff-Australia}
O.~M. Cliff, N.~Harding, M.~Piraveenan, E.~Y. Erten, M.~Gambhir, and
  M.~Prokopenko, ``Investigating spatiotemporal dynamics and synchrony of
  influenza epidemics in australia: An agent-based modelling approach,''
  \emph{Simulation Modelling Practice and Theory}, vol.~87, pp. 412--431, 2018.

\bibitem{Zhang1}
H.~Zhang, J.~Zhang, C.~Zhou, M.~Small, and B.~Wang, ``Hub nodes inhibit the
  outbreak of epidemic under voluntary vaccination,'' \emph{New Journal of
  Physics}, vol.~12, no.~2, p. 023015, 2010.

\bibitem{Chang}
S.~L. Chang, M.~Piraveenan, and M.~Prokopenko, ``Impact of network
  assortativity on epidemic and vaccination behaviour,'' \emph{arXiv preprint
  arXiv:2001.01852}, 2020.

\bibitem{Bagnoli}
F.~Bagnoli, P.~Lio, and L.~Sguanci, ``Risk perception in epidemic modeling,''
  \emph{Physical Review E}, vol.~76, no.~6, p. 061904, 2007.

\bibitem{Bauch1}
C.~T. Bauch and D.~J. Earn, ``Vaccination and the theory of games,''
  \emph{Proceedings of the National Academy of Sciences}, vol. 101, no.~36, pp.
  13\,391--13\,394, 2004.

\bibitem{Bauch2}
C.~T. Bauch, A.~P. Galvani, and D.~J. Earn, ``Group interest versus
  self-interest in smallpox vaccination policy,'' \emph{Proceedings of the
  National Academy of Sciences}, vol. 100, no.~18, pp. 10\,564--10\,567, 2003.

\bibitem{Reluga1}
T.~C. Reluga, C.~T. Bauch, and A.~P. Galvani, ``Evolving public perceptions and
  stability in vaccine uptake,'' \emph{Mathematical biosciences}, vol. 204,
  no.~2, pp. 185--198, 2006.

\bibitem{Reluga2}
T.~C. Reluga and A.~P. Galvani, ``A general approach for population games with
  application to vaccination,'' \emph{Mathematical biosciences}, vol. 230,
  no.~2, pp. 67--78, 2011.

\bibitem{Zhang2}
H.~Zhang, F.~Fu, W.~Zhang, and B.~Wang, ``Rational behavior is a
  ‘double-edged sword’when considering voluntary vaccination,''
  \emph{Physica A: Statistical Mechanics and its Applications}, vol. 391,
  no.~20, pp. 4807--4815, 2012.

\bibitem{Fine-Clarkson}
P.~E.~M. Fine and J.~A. Clarkson, ``Individual versus public priorities in the
  determination of optimal vaccination policies,'' \emph{American journal of
  epidemiology}, vol. 124, no.~6, pp. 1012--1020, 1986.

\bibitem{Kremer}
M.~Kremer, ``Integrating behavioral choice into epidemiological models of
  aids,'' \emph{The Quarterly Journal of Economics}, vol. 111, no.~2, pp.
  549--573, 1996.

\bibitem{Vardavas}
R.~Vardavas, R.~Breban, and S.~Blower, ``Can influenza epidemics be prevented
  by voluntary vaccination?'' \emph{PLoS computational biology}, vol.~3, no.~5,
  2007.

\bibitem{Del_Valle}
S.~Del~Valle, H.~Hethcote, J.~M. Hyman, and C.~Castillo-Chavez, ``Effects of
  behavioral changes in a smallpox attack model,'' \emph{Mathematical
  biosciences}, vol. 195, no.~2, pp. 228--251, 2005.

\bibitem{Chen2}
F.~H. Chen, ``Rational behavioral response and the transmission of stds,''
  \emph{Theoretical population biology}, vol.~66, no.~4, pp. 307--316, 2004.

\bibitem{Funk-Review}
S.~Funk, M.~Salath{\'e}, and V.~A. Jansen, ``Modelling the influence of human
  behaviour on the spread of infectious diseases: a review,'' \emph{Journal of
  the Royal Society Interface}, vol.~7, no.~50, pp. 1247--1256, 2010.

\bibitem{Reluga3}
T.~C. Reluga, ``Game theory of social distancing in response to an epidemic,''
  \emph{PLoS computational biology}, vol.~6, no.~5, 2010.

\bibitem{Poletti1}
P.~Poletti, M.~Ajelli, and S.~Merler, ``The effect of risk perception on the
  2009 h1n1 pandemic influenza dynamics,'' \emph{PloS one}, vol.~6, no.~2,
  2011.

\bibitem{Poletti2}
------, ``Risk perception and effectiveness of uncoordinated behavioral
  responses in an emerging epidemic,'' \emph{Mathematical Biosciences}, vol.
  238, no.~2, pp. 80--89, 2012.

\bibitem{Poletti3}
P.~Poletti, B.~Caprile, M.~Ajelli, A.~Pugliese, and S.~Merler, ``Spontaneous
  behavioural changes in response to epidemics,'' \emph{Journal of theoretical
  biology}, vol. 260, no.~1, pp. 31--40, 2009.

\bibitem{Funk1}
S.~Funk, E.~Gilad, C.~Watkins, and V.~A. Jansen, ``The spread of awareness and
  its impact on epidemic outbreaks,'' \emph{Proceedings of the National Academy
  of Sciences}, vol. 106, no.~16, pp. 6872--6877, 2009.

\bibitem{Chen1}
F.~H. Chen, ``Modeling the effect of information quality on risk behavior
  change and the transmission of infectious diseases,'' \emph{Mathematical
  biosciences}, vol. 217, no.~2, pp. 125--133, 2009.

\bibitem{d'Onofrio}
A.~d’Onofrio and P.~Manfredi, ``Information-related changes in contact
  patterns may trigger oscillations in the endemic prevalence of infectious
  diseases,'' \emph{Journal of Theoretical Biology}, vol. 256, no.~3, pp.
  473--478, 2009.

\bibitem{Fu}
F.~Fu, D.~I. Rosenbloom, L.~Wang, and M.~A. Nowak, ``Imitation dynamics of
  vaccination behaviour on social networks,'' \emph{Proceedings of the Royal
  Society B: Biological Sciences}, vol. 278, no. 1702, pp. 42--49, 2011.

\bibitem{van_Boven}
M.~van Boven, D.~Klinkenberg, I.~Pen, F.~J. Weissing, and H.~Heesterbeek,
  ``Self-interest versus group-interest in antiviral control,'' \emph{PLoS
  One}, vol.~3, no.~2, 2008.

\bibitem{Philipson1}
T.~Philipson and R.~Posner, \emph{Private choices and public health: The AIDS
  epidemic in an economic perspective}.\hskip 1em plus 0.5em minus 0.4em\relax
  Harvard University Press, 1993.

\bibitem{Philipson2}
P.-Y. Geoffard and T.~Philipson, ``Rational epidemics and their public
  control,'' \emph{International economic review}, pp. 603--624, 1996.

\bibitem{Philipson3}
------, ``Disease eradication: private versus public vaccination,'' \emph{The
  American Economic Review}, vol.~87, no.~1, pp. 222--230, 1997.

\bibitem{Game_survey}
S.~L. Chang, M.~Piraveenan, P.~Pattison, and M.~Prokopenko, ``Game theoretic
  modelling of infectious disease dynamics and intervention methods: a
  review,'' \emph{Journal of Biological Dynamics}, vol.~14, no.~1, pp. 57--89,
  2020.

\bibitem{Gross}
T.~Gross, C.~J.~D. D’Lima, and B.~Blasius, ``Epidemic dynamics on an adaptive
  network,'' \emph{Physical review letters}, vol.~96, no.~20, p. 208701, 2006.

\bibitem{Shaw}
L.~B. Shaw and I.~B. Schwartz, ``Fluctuating epidemics on adaptive networks,''
  \emph{Physical Review E}, vol.~77, no.~6, p. 066101, 2008.

\bibitem{Zanette}
D.~H. Zanette and S.~Risau-Gusm{\'a}n, ``Infection spreading in a population
  with evolving contacts,'' \emph{Journal of biological physics}, vol.~34, no.
  1-2, pp. 135--148, 2008.

\bibitem{Somarakis}
C.~E. Somarakis, G.~P. Papavassilopoulos, and F.~E. Udwadia, ``Nonlinear
  dynamics of a nwe cellular automata model,'' \emph{ECMS Conf., Nicosia,
  Cyprus}, June 2008.

\bibitem{random}
P.~Erd{\H{o}}s and A.~R{\'e}nyi, ``On the evolution of random graphs,''
  \emph{Publ. Math. Inst. Hung. Acad. Sci}, vol.~5, no.~1, pp. 17--60, 1960.

\bibitem{Scale_free}
A.-L. Barab{\'a}si, R.~Albert, and H.~Jeong, ``Scale-free characteristics of
  random networks: the topology of the world-wide web,'' \emph{Physica A:
  statistical mechanics and its applications}, vol. 281, no. 1-4, pp. 69--77,
  2000.

\bibitem{renardy}
M.~Renardy and R.~C. Rogers, \emph{An introduction to partial differential
  equations}.\hskip 1em plus 0.5em minus 0.4em\relax Springer Science \&
  Business Media, 2006, vol.~13.

\bibitem{Barabasi}
A.-L. Barab{\'a}si \emph{et~al.}, \emph{The scale free property, Chapt. 4, in
  Network science}.\hskip 1em plus 0.5em minus 0.4em\relax Cambridge university
  press, 2016.

\bibitem{small_world}
D.~J. Watts and S.~H. Strogatz, ``Collective dynamics of
  ‘small-world’networks,'' \emph{nature}, vol. 393, no. 6684, p. 440, 1998.

\end{thebibliography}

\end{document}